\documentclass[12pt]{article}

\usepackage[margin=2.5cm]{geometry}
\usepackage{amsthm,amsmath,amsfonts,amssymb}
\usepackage{graphicx,float}
\usepackage{multirow,setspace}
\usepackage[authoryear]{natbib}
\usepackage{enumerate}
\usepackage{caption}
\usepackage{subcaption}
\usepackage{bm}
\usepackage{color}
\usepackage{url}
\usepackage[colorlinks,citecolor=blue,urlcolor=blue]{hyperref}
\usepackage[format=plain, labelfont={bf,it}, textfont=it]{caption}
\theoremstyle{plain}

\newtheorem{theorem}{Theorem}[section]

 \newtheorem{assumption}{Assumption}[section] % Numbered by section
    \newtheorem{proposition}[assumption]{Proposition} % Shares numbering with assumption
    
\theoremstyle{remark}

\newtheorem{remark}[theorem]{Remark}
\title{\vspace{-40pt}Sign and signed rank tests for paired functions\vspace{-10pt}}
\author{Mark J. Meyer \\ Department of Mathematics and Statistics, Georgetown University \\ Washington, DC USA}
\date{}

\begin{document}

	\maketitle
	
\begin{abstract}
Simple nonparametric tests for paired functional data are an understudied area, despite recent advances in similar tests for other types of functional data. While the sign test has received limited treatment, the signed rank-type test has not previously been examined. The aim of the present work is to develop and evaluate these types of tests for functional data. We derive a simple, theoretical framework for both sign and signed rank tests for pairs of functions. In particular, we demonstrate that doubly ranked testing---a newly developed framework for testing hypotheses involving functional data---is a useful conduit for examining hypotheses regarding pairs o,f functions. We briefly examine the operating characteristics of all derived tests. We also use the described approaches to re-analyze pairs of functions from a randomized crossover study of heart health during simulated flight.
\end{abstract}

	\section{Introduction}
	\label{s:intro}
	
	Recent work on nonparametric tests for functional data considers a variety of approaches for analyzing two or more independent samples including Mann-Whitney or Wilcoxon rank sum type tests, Anderson-Darling tests, permutation tests, and Kruskal-Wallis tests \citep{Hall2007, Lopez2009, Lopez2010, Chak2015, PomannStaicu2016, Lopez2017, Abramowicz2018,Melendez2021,Meyer2025}. However, the nonparametric paired functional case remains limitedly studied with work by \cite{Berrett2021} and \cite{Melendez2021} specifically addressing it. For the purposes of this manuscript, we assume that the data are curves which have been sampled on a finite grid. While functional data can take the form of higher dimensional objects, such as surfaces, we limit our consideration to curves.
	
	\cite{Melendez2021} consider random projection-based sign tests. The test begins by projecting the curve into a randomly chosen basis space constructed using Brownian motion. The authors numerically integrate over each to obtain a pair-specific score. The scores are then analyzed using a sign test---see \cite{Kloke2015}, among others, for a discussion of the sign test. \cite{Berrett2021} consider independence tests to determine if two samples are dependent or not. This method can be used to assess dependency, however, it is a different testing framework than the usual paired data-type test. The null of their test for independence is that the samples are independent, i.e. that they come from separate populations. Our setting assumes a dependency in the data whereby functions are sampled twice on each subject and seeks to account for it. The goal is then to assess if there is a difference between the conditions under each measurement occurrence.
	
%	 \cite{Hall2007, Lopez2009, Lopez2010, Chak2015, PomannStaicu2016, Lopez2017, Abramowicz2018, Berrett2021, Melendez2021}, and \cite{Meyer2025}. 
	
	%	when the data are either curves or surfaces
	
%	a random projection-based test by \cite{Melendez2021}, maps the points from the high dimensional functional space to a randomly chosen low-dimensional space defined using Brownian motion. This approach has application in parametric tests for grouped functional as well, see for example \cite{Cuesta2010}. The random projections-based approach effectively generates scores for each subject via a random basis function and an integral approximation of the subject-specific curve. It then treats these scores as the data in a traditional MWW test, thus the scores are ranked and then summed by group. Neither of these alternative methods have publicly available code.
	
%	random projection-based approaches perform well in larger samples \citep{Melendez2021}.  the brownian motion-based basis functions for the random projections. When using random projections, the integration is performed before constructing ranks which ignores the fact that rank may be dynamic over time. To handle the curve ranking problem,  rely on subject-specific summary scores, although neither constructs their scores under the null. Thus, the tests are only conducted under the null when comparing the scores between groups. 
%	
%	\cite{Berrett2021}

	%%% EDIT %%%
	In this work, we describe two sign tests
%	---one similar to a test described by \cite{Melendez2021}---
	and a signed doubly ranked test for pairs of functional data. Our functional sign tests are constructed under the null based on a series of simple assumptions with the goal of constructing a pair-specific univariate score for the difference in curves. We consider both a sufficient statistic-based sign test and an integral-based sign test similar in spirit to the work of \cite{Melendez2021}. The signed doubly ranked test builds upon the work of \cite{Meyer2025} who develop the framework for two (or more) independent sets of functions. Doubly ranked tests are nonparametric tests that first rank functional data, preprocessed or raw, at each time point and then construct a summary rank using a relevant statistic. They then re-rank the summaries in the final step of the analysis. The extension of the framework to the paired case is not immediate and, consequently, has not previously been considered. Thus, we undertake a study of how to implement a doubly ranked test for paired functions. Via an empirical study, we show that while all tests perform reasonably well, the signed doubly ranked test performs the best in terms of type I error and power. We illustrate the use of all of the tests using a crossover study of in-flight heart rate and heart rate variability first described by \cite{Meyer2019}. 
	
	The remainder of this manuscript proceeds as follows: in Section~\ref{s:paired} we formalize the concept of paired functions, using the flight study as motivation. Section~\ref{s:sign} contains our discussion of the functional sign tests while Section~\ref{s:sdrt} provides the derivation of the signed doubly ranked test. Our empirical study is in Section~\ref{s:sim} and our reanalysis of the flight study is in Section~\ref{s:data}. Finally, we provide a brief discussion of the various approaches to nonparamterically analyzing paired functions in Section~\ref{s:disc}.

	\section{Paired Functions}
	\label{s:paired}
	
	In general, functional data is data where the unit of observation is a curve or surface measured over an interval (or intervals) of time or location. Paired functional data is the case where two curves or surfaces are measured in sets on the same subjects perhaps under different experimental conditions or at different calendar times. Restricting our focus to functions, let $X_{ij}(s)$ denote the subject $i$'s, $i = 1, \ldots, n$, function at time or location $s$ for condition $j$, $j = 0,1$. Further, let $D_i(s) = X_{i1}(s) - X_{i0}(s)$ be the difference function for subject $i$. That is, it is the function that tracks the point-wise difference between a subject's functions over $s \in \mathcal{S}$.
	
	While the functions are written as as continuous functions, we do not observe them continuously. Instead, we observe discrete realizations of $X_j(s)$, typically over the grid $\mathcal{S} = \left[s: s = s_1, \ldots, s_S\right]$. To adjust for this discrepancy between what we are trying to measure and what is physically measurable, the observed data is usually projected into a well-behaved basis space to model the underlying function. There are many basis functions available for preprocessing, see \cite{Wang2016} among others. We consider using a functional principal components analysis (FPCA) approach using Fast Covariance Estimation (FACE) which is a nonparametric FPCA \citep{Xiao2013,Xiao2016} and FPCA with smoothed covariance or FPCA SC \citep{Yao2005,Di2009}.
	
	One example of paired functional data comes from an experiment described by \cite{Meyer2019} which examines indicators of heart health under two conditions: at an altitude of 7000 ft. (flight) and at sea level (control). The experiment was designed to mimic a six hour flight in a commercial airline. Conditions were controlled via the use of a hypobaric pressure chamber and subjects were blinded to the condition. Various measurements were taken every five minutes during flight. Preprocessed and differenced data from the five metrics considered by this study is depicted in Figure~\ref{f:hr}. Under each condition, heart rate and heart rate variability were measured pseudo-continuously for the duration of the experiment, about 400 minutes. The measurements were aggregated to the five-minute scale which is how the data is presented in Figure~\ref{f:hr}. Due to missingness, we use FPCA SC to preprocess and retain components coressponding to 99\% of the explained variability.
	
%	\begin{figure}
%		\centering
%		\includegraphics[width = 2.1in, height = 2.1in]{hr_diff.eps}
%		\caption{Difference in heart rate (flight condition $-$ control condition).\label{f:hr}}
%	\end{figure}

	\begin{figure}
		\centering
		\includegraphics[width = 2.1in, height = 2.1in]{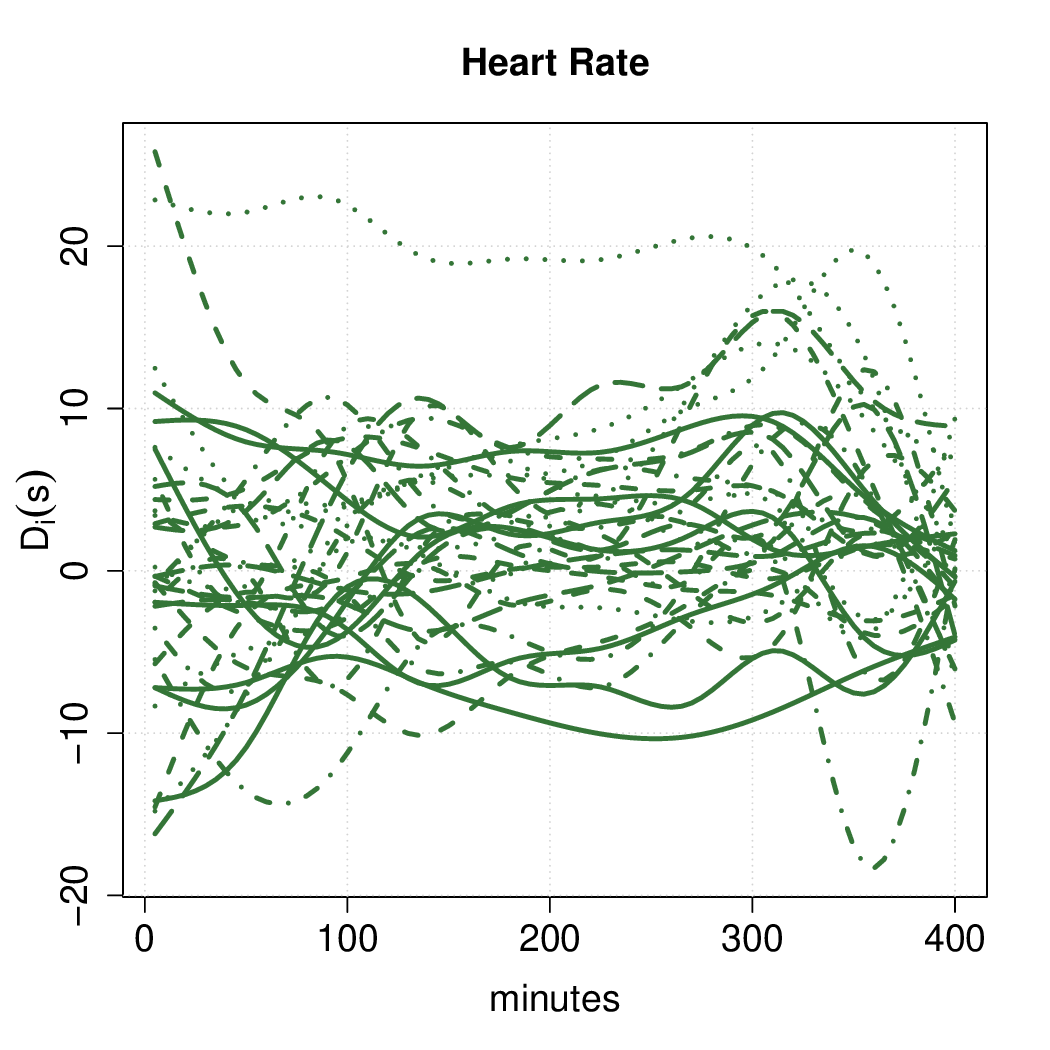}
		\includegraphics[width = 2.1in, height = 2.1in]{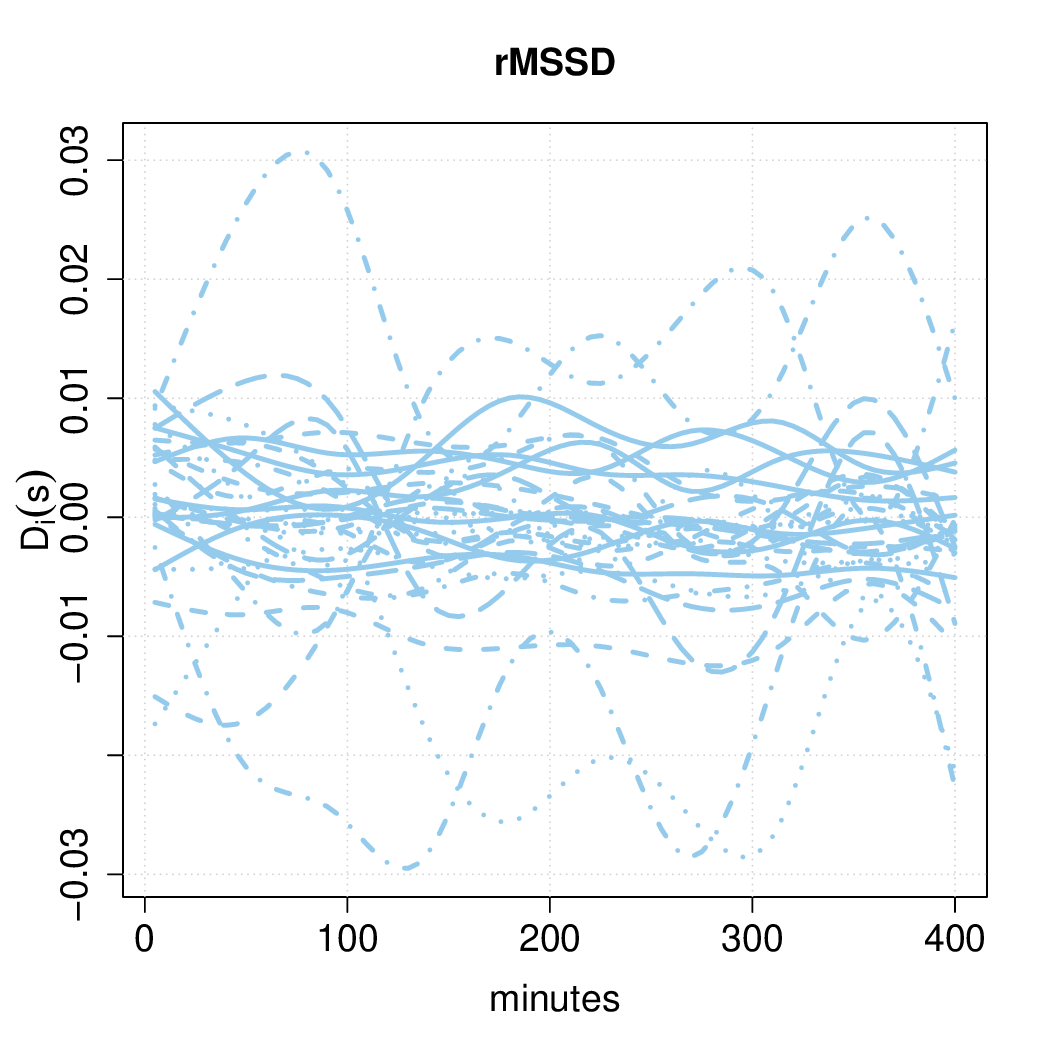}
		\includegraphics[width = 2.1in, height = 2.1in]{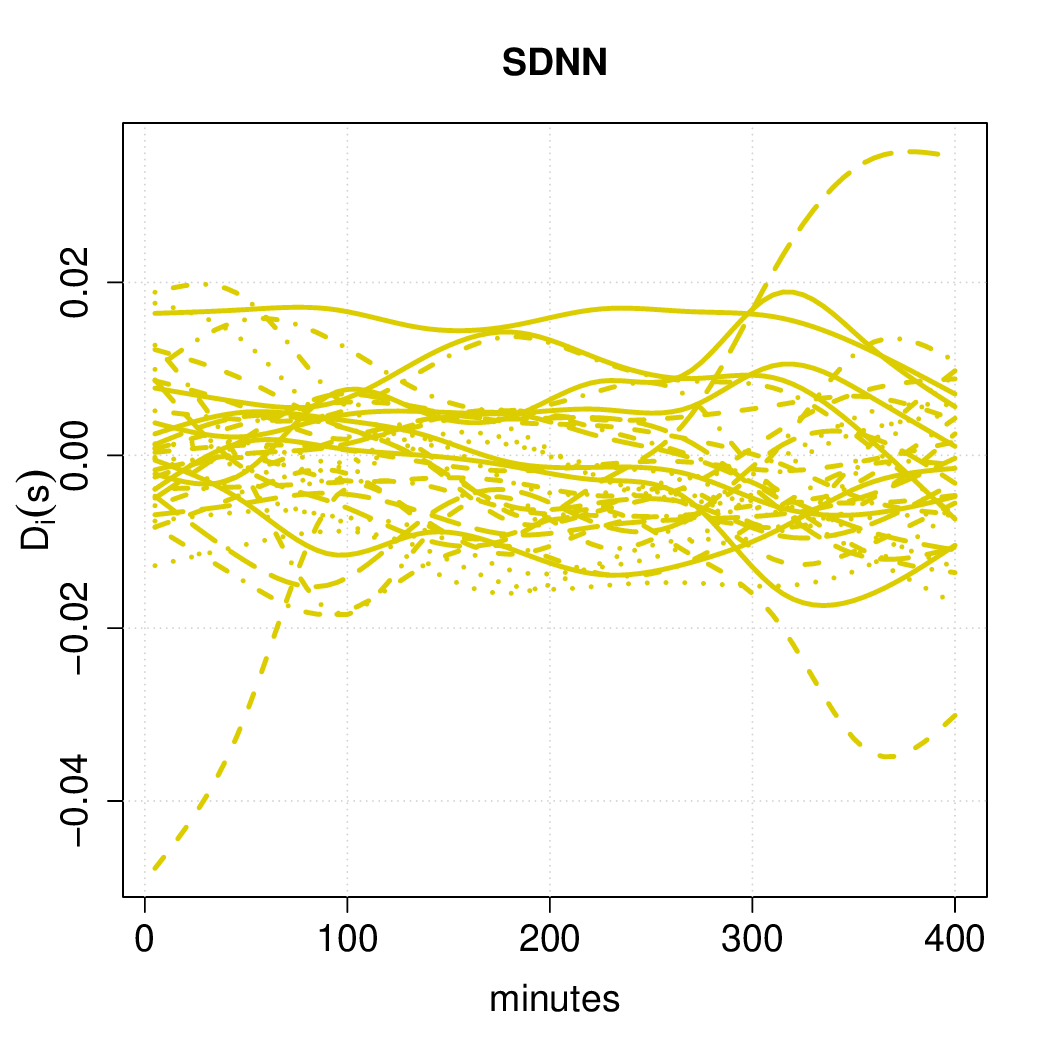}
		\includegraphics[width = 2.1in, height = 2.1in]{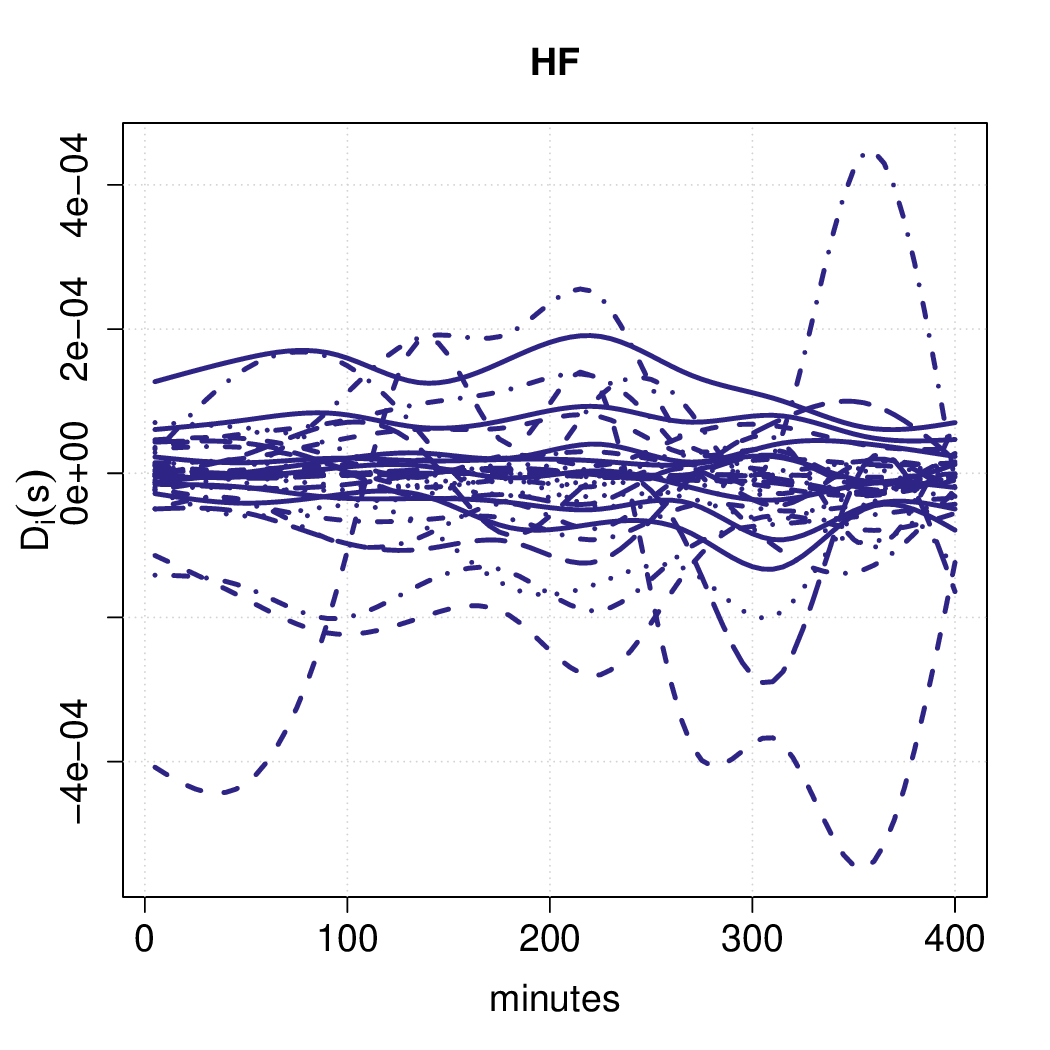}
		\includegraphics[width = 2.1in, height = 2.1in]{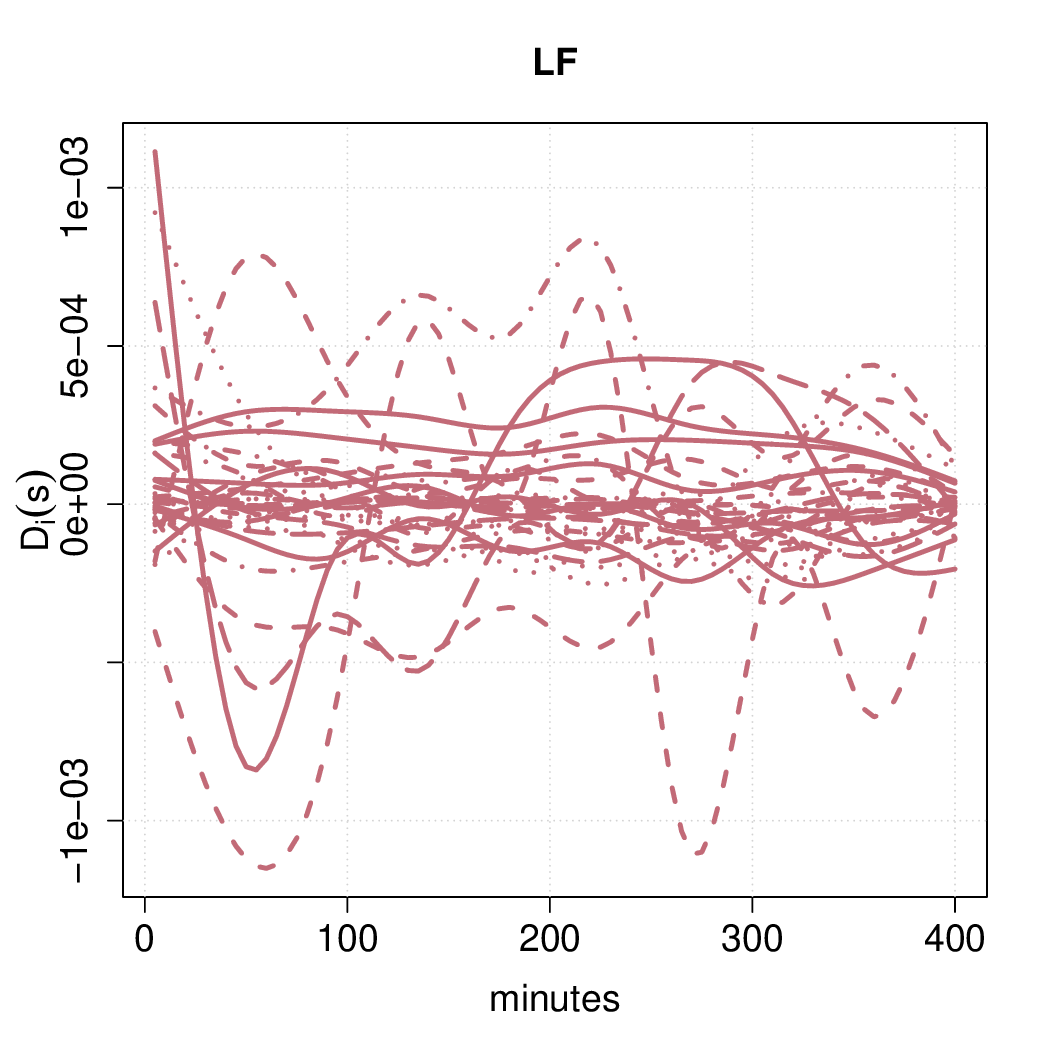}
		\caption{Preprocessed and differenced data from \cite{Meyer2019} (flight condition $-$ control condition). rMSSD (root mean square of successive RR interval differences), SDNN (the standard deviation of normal-to-normal intervals), HF (high frequency power), and LF (low frequency power) all measure heart rate variability. \label{f:hr}}
	\end{figure}

	\section{Functional Sign Tests}
	\label{s:sign}
	
%	Reorder: procedure first with general statistic then suff stat then average rank

%Change ties to zeros

	Under $H_0$, we assume no difference between measurement occurrences at each time $s$. If $\Delta(s)$ is the true difference, then $\Delta(s) \stackrel{H_0}{=} 0$. If the null is true for all $s$ and we treat the $D_i(s)$ as a sample of differences over time, then we should see no difference for summaries of the $D_i(s)$ values under $H_0$ as well. Let $\Lambda$ denote this global difference which also represents a true difference in the curves over $s$. Then, under $H_0$, $\Lambda = 0$ versus the alternative that $H_A : \Lambda \neq 0$. To evaluate $H_0$ in a nonparametric fashion, we first examine the signs of the values of $D_i(s)$ and then construct relevant summary statistics.
		
	Let $G_i(s)$ denote the sign of the difference for subject $i = 1, \ldots, n$ at time or location $s$, thus $G_i(s) = \text{sign}[D_i(s)]$. In the absence of zeros, $G_i(s)$ follows a generalized Rademacher distribution of the form
	\begin{align*}
		P_{G_i(s)}[g_i(s)] = \left\{ \begin{array}{ccl}
					1 - p & & g_i(s) = -1 \\
					p & & g_i(s) = 1\\
					\end{array}\right.,
	\end{align*}
	where $g_i(s)$ is a realization of $G_i(s)$ and $p \in (0, 1)$. We assume there is some true $p$, potentially subject-specific, that dictates the sign at each time point $s$. Define the statistic $t_r[g_i(\mathcal{S})] = \sum_{s = 1}^S g_i(s)$.
	
\begin{proposition}\label{prop:suffRad}
$t_r[g_i(\mathcal{S})]$ is sufficient for $p$.
\end{proposition}

\begin{proof}
Let $B_i(s) = 2G_i(s) - 1$. Since $G_i(s)$ follows a Rademacher distribution, $B_i(s)$ is Bernoulli with probability of success $p$. A realization of $B_i(s)$, $b_i(s)$, is sufficient for $p$. Since realizations of $G_i(s)$, $g_i(s)$, are functions of $b_i(s)$, $g_i(s)$ is also sufficient for $p$ at time or location $s$. Summing over $s$, $t_r[g_i(\mathcal{S})]$ is sufficient for $p$ as a function of sufficient statistics.
\end{proof}

	 If zeros are allowed, $G_i(s)$ follows a categorical distribution with probability vector $\textbf{p} = \left( \begin{array}{ccc} p_{-1} & p_{0}& p_1\end{array}\right)'$. We assume that $\textbf{p}$ is some true vector of probabilities, potentially subject-specific again, that dictates the sign at each time point. A sufficient statistic for $\textbf{p}$ is the vector
	 \begin{align*}
	 	\textbf{g}_i(s) = 
		\left(\begin{array}{c}
						\left[g_i(s) = -1\right] \\
						\left[g_i(s) = 0\right] \\
						\left[g_i(s) = 1\right] \\
			\end{array}\right),
	 \end{align*}
	 where $[\cdot]$ denotes the Iverson bracket, thus $[A] = 1$ if $A$ is true and 0 otherwise. Define the following statistic:
	 \begin{align*}
	 	t_c[g_i(\mathcal{S})] = \sum_{s = 1}^S \textbf{w}'
			\left(\begin{array}{c}
						\left[g_i(s) = -1\right] \\
						\left[g_i(s) = 0\right] \\
						\left[g_i(s) = 1\right] \\
			\end{array}\right),
	 \end{align*}
	 where $\textbf{w}$ is a vector of real-valued weights.

\begin{proposition}\label{prop:suffCat}
$t_c[g_i(\mathcal{S})]$ is sufficient for $\textbf{p}$.
\end{proposition}

\begin{proof}
$t_c[g_i(\mathcal{S})]$ is a linear combination of statistics that are sufficient for $\textbf{p}$. Thus, $t_c[g_i(\mathcal{S})]$ is also sufficient for $\textbf{p}$.
\end{proof}

	\begin{remark}\label{rm:sign}
		If $\textbf{w} =  \left(\begin{array}{c}
						-1 \\ 0 \\ 1 \\
			\end{array}\right)$, then $t_c[g_i(\mathcal{S})]$ $= t_r[g_i(\mathcal{S})]$.
	\end{remark}
	
	Using these weights, the statistics will be the same and consist of summing the signs of the difference function over $\mathcal{S}$. We then calculate $t_r[g_i(\mathcal{S})]$ for each subject resulting in a univariate score and apply the univariate sign test to the sample of $t_r[g_i(\mathcal{S})]$. The test statistic is
	\begin{align*}
		U_r^+ = \sum_{i=1}^n 1\left[ t_r\{g_i(\mathcal{S})\} > 0 \right],
	\end{align*}
	which, under $H_0$, follows a Binomial distribution with probability of success equal to $1/2$.
	
	An alternative univariate score construction involves integrating over the difference function for each subject. For realizations of $D_i(s)$, $d_i(s)$, the integral-based statistic is
	\begin{align*}
		t_n[d_i(\mathcal{S})] = \int_{\mathcal{S}} d_i(s) ds.
	\end{align*}
	Since $d_i(s)$ is observed discretely and $\mathcal{S}$ is measured as grid, calculating $t_n[d_i(\mathcal{S})]$ requires numerical integration. We implement the trapezoidal rule to approximate the integral with grid points determined by the sampling grid, $\mathcal{S}$. \cite{Melendez2021} propose a similar statistic, integrating $X_{i1}(s)$ and $X_{i0}(s)$ separately before differencing. The key distinctions between their work and ours is the type of basis function and order in which the basis transformation is applied. \cite{Melendez2021} perform their Brownian motion-based basis transformation on the measurements from each condition separately which effectively assumes a difference in the measurements---that is, preprocessing in this manner is under the alternative. We perform an FPCA on the joint sample. By ignoring condition in our prepreprocessing, we are performing the basis transformation assuming no difference, i.e. under the null. Regardless, the univariate sign test is then applied to the summary statistic $t_n[d_i(\mathcal{S})]$ resulting in a test statistic of
	\begin{align*}
		U_n^+ = \sum_{i=1}^n 1\left[ t_n\{g_i(\mathcal{S})\} > 0 \right],
	\end{align*}
	which has the same null distribution as $U_r^+$.

	\section{Signed Doubly Ranked Tests}
	\label{s:sdrt}
	
	As in the univariate case, the functional sign tests in Section~\ref{s:sign} do not account for the magnitude of the difference, only the sign of the difference. To accommodate for the magnitude, we now extend the doubly ranked test, first proposed by \cite{Meyer2025} for independent samples, to the dependent samples (or paired) case. The null and alternative for signed doubly ranked tests is the same as in the functional sign test. Thus, $H_0 : \Lambda = 0$ versus $H_A : \Lambda \neq 0$. To evaluate $H_0$, we first preprocess the functions ignoring measurement occurrence and construct the differences, $d_i(s)$; second, we rank and sign the $d_i(s)$ at each $s$; third, we summarize the signed ranks over time with a statistic, $t[d_i(\mathcal{S})]$; and forth, we perform the Wilcoxon signed rank test on the summaries from step three. Since the last step ranks the summarized ranks, this test falls within the doubly ranked testing framework. 
	
%	In Step 2), we rank the absolute value of $d_i(s)$ so that the ranking procedure is consistent with the signed rank test.
	
	Let $\mathcal{U}(a,b)$ denote the discrete uniform distribution on the interval $[a,b]$, $\mathcal{R}(p)$ denote the Rademacher distribution with probability $p$, $\mathcal{C}(\textbf{p})$ denote the categorical distribution with probability vector $\textbf{p}$, and $R|x|$ denote the rank of the absolute value of $x$ with respect to the sample of $x$'s. At a specific value of $s$, we make the following set of assumptions:
	\begin{assumption}\label{a:unif}
		$R|d_{i}(s)| \stackrel{H_0}{\sim} \mathcal{U}\{1,n\}$.
	\end{assumption}

	\begin{assumption}\label{a:rad}
		In the absence of zeros, sign$[d_i(s)]$ $\stackrel{H_0}{\sim} \mathcal{R}(1/2)$.
	\end{assumption}

	\begin{assumption}\label{a:cat}
		In the presence of zeros, sign$[d_i(s)]$ $\stackrel{H_0}{\sim} \mathcal{C}(\textbf{p})$ where $\textbf{p} = \left( \begin{array}{ccc} \frac{1}{3} & \frac{1}{3} & \frac{1}{3} \end{array}\right)'$.
	\end{assumption}

	\begin{assumption}\label{a:ind}
		Under $H_0$, sign$[d_i(s)] \perp R|d_{i}(s)|$.
	\end{assumption}
	
	At a fixed $s$, the observed $d_i(s)$ is a scalar measurement.	Similar assumptions are made, specifically Assumptions~\ref{a:unif}, ~\ref{a:rad}, and~\ref{a:ind}, in the univariate signed rank test. Given these assumptions, we make the following remarks:

	\begin{remark}\label{rm:exr}
		Under $H_0$ and in the absence of zeros, $E\left[ \text{sign}\{d_i(s)\} \right] = 0$.
	\end{remark}

	\begin{remark}\label{rm:exc}
		Under $H_0$ and in the presence of zeros, $E\left[ \text{sign}\{d_i(s)\} \right] = \textbf{p}$.
	\end{remark}

	Remarks~\ref{rm:exr} and~\ref{rm:exc} are direct results of Assumptions~\ref{a:rad} and~\ref{a:cat}, respectively. To summarize the signed ranks across $\mathcal{S}$, we propose the following statistics:
\begin{align}
	t_d[d_i(\mathcal{S})] &= \frac{1}{S} \sum_{s=1}^S \text{sign}[d_i(s)] R|d_i(s)| \text{ and}\label{eq:nozeros}\\
	t^{\dagger}_d[d_i(\mathcal{S})] &= \frac{1}{S} \sum_{s=1}^S \textbf{w}'\text{sign}\{d_i(s)\} R|d_i(s)|,\label{eq:zeros}
\end{align}
where $t_d[d_i(\mathcal{S})]$ assumes no zeros and $t^{\dagger}_d[d_i(\mathcal{S})]$ allows for zeros to occur. The vector $\textbf{w}$ contains real-valued weights. 

	\begin{remark}\label{rm:wts}
		Setting $\textbf{w} =  \left(\begin{array}{c}
						-1 \\ 0 \\ 1 \\
		\end{array}\right)$ corresponds to setting the weights equal to the support of sign$[d_i(s)]$ in the presence of zeros.
	\end{remark}

	\begin{proposition}\label{prop:ext}
		Under $H_0$, $E\left[t_d\left\{d_i(\mathcal{S})\right\}\right] = 0$.
	\end{proposition}
	
	\begin{proof}
	Via Assumptions~\ref{a:rad} and~\ref{a:ind} as well as Remark~\ref{rm:exr}, we have
\begin{align*}
	E\left[t_d\left\{d_i(\mathcal{S})\right\}\right] 
%	&= E\left[ \frac{1}{S} \sum_{s=1}^S \text{sign}\{d_i(s)\} R|d_i(s)| \right]\\
	&=  \frac{1}{S} \sum_{s=1}^S E\left[ \text{sign}\{d_i(s)\} R|d_i(s)| \right]\\
	&=  \frac{1}{S} \sum_{s=1}^S E\left[ \text{sign}\{d_i(s)\}\right] E\left[  R|d_i(s)| \right]\\
	&= 0 
\end{align*}		
	\end{proof}

	\begin{proposition}\label{prop:extt}
		Under $H_0$, $E\left[t^{\dagger}_d\left\{d_i(\mathcal{S})\right\}\right] = 0$.
	\end{proposition}

	\begin{proof}
	Via Assumptions~\ref{a:cat} and~\ref{a:ind} as well as Remarks~\ref{rm:exc} and~\ref{rm:wts}, we have
\begin{align*}
	E\left[t^{\dagger}_d\left\{d_i(\mathcal{S})\right\}\right] 
%	&= E\left[ \frac{1}{S} \sum_{s=1}^S \text{sign}\{d_i(s)\} R|d_i(s)| \right]\\
%	&= \frac{1}{S} \sum_{s=1}^S E\left[\textbf{w}'\text{sign}\{d_i(s)\}  R|d_i(s)| \right]\\
%	&= \frac{1}{S} \sum_{s=1}^S E\left[\textbf{w}'\text{sign}\{d_i(s)\} \right] E\left[ R|d_i(s)| \right]\\
	&= \frac{1}{S} \sum_{s=1}^S \textbf{w}' E\left[\text{sign}\{d_i(s)\} \right] E\left[ R|d_i(s)| \right]\\
	&= \frac{1}{S} \sum_{s=1}^S \textbf{w}' \textbf{p } E\left[ R|d_i(s)| \right]\\
	&= 0 
\end{align*}
		The final step holds since the inner product of $\textbf{w}$ and $\textbf{p}$ is zero when $\textbf{w}$ is as defined in Remark~\ref{rm:wts} and $\textbf{p}$ is as defined in Assumption~\ref{a:cat}.
	\end{proof}
	
	Thus, both summaries maintain the properzeros of the signed ranks under $H_0$. These summaries become the data in Step 4) of the doubly ranked procedure. Since they are one-sample summaries, we apply the Wilcoxon signed rank test. The signed doubly ranked test statistics are then either
\begin{align*}
	W = &\frac{1}{2} \sum_{i=1}^n \text{sign}\left[ t_d\left\{d_i(\mathcal{S})\right\} \right] R\left| t_d\left\{d_i(\mathcal{S})\right\} \right| \text{ or}\\
	W^{\dagger} = &\frac{1}{2} \sum_{i=1}^n \text{sign}\left[ t^{\dagger}_d\left\{d_i(\mathcal{S})\right\} \right] R\left| t^{\dagger}_d\left\{d_i(\mathcal{S})\right\} \right|.
%	W^+_{f,r} = &\frac{1}{2} \sum_{i=1}^n \text{sign}\left[t_{f,r}\left\{z_{i}(\mathcal{S})\right\}\right] R\left| t_{f,r}\left[z_{i}(\mathcal{S})\right] \right|  + \frac{n(n+1)}{4},
%	\text{ or}\\
%	W^+_{f,n} = &\frac{1}{2} \sum_{i=1}^n \text{sign}\left[t_{f,n}\left\{z_{i}(\mathcal{S})\right\}\right] R\left| t_{f,n}\left[z_{i}(\mathcal{S})\right] \right| \\
%	&+ \frac{n(n+1)}{4},
\end{align*}
The distributions of $W$ and $W^{\dagger}$ are centered at zero. Since both $t_d\left\{d_i(\mathcal{S})\right\}$ and $t^{\dagger}_d\left\{d_i(\mathcal{S})\right\}$ are scalar summaries of pair-specific ranks, $W$ and $W^{\dagger}$ have the same distribution as the Wilcoxon signed rank test \citep{Wilcoxon1945}.

	\section{Empirical Study}
	\label{s:sim}
	
	We examine both a general empirical study and one motivated by our data illustration. In both, we consider samples of size $n = 15, 30,$ and $60$. We assume the grid lies on the unit interval and take sampling densities of $S = 40, 120,$ and 360 for the general study while using $S = 80$ for the data-based study. To generate pairs of functions, we modify the approach taken by  \cite{Chak2015}, \cite{Berrett2021}, and \cite{Meyer2025} who generate data using a Karhunen-Lo\`eve expansion similar to
	\begin{align*}
		X_{ij}(s) &= \sum_{k=1}^K \sqrt{2} \left[(k - 0.5)\pi \right]^{-1} Z_{ijk} \sin\left[(k - 0.5)\pi s\right],
	\end{align*}
	for $i = 1, \ldots, n$ and $j = 0,1$. Given a single subject, $[\begin{array}{cc} Z_{i0k} & Z_{i1k} \end{array}]'$ is a $2\times1$ vector which we draw from either a multivariate normal or a multivariate $t$ distribution with two degrees of freedom, both centered at $[\begin{array}{cc} 0 & 0 \end{array}]'$. To induce dependency, we set the scale matrix to $\left[ \begin{array}{cc}
				1 & \rho \\
				\rho & 1
			\end{array}
		\right],$
	where $\rho = 0.5$ or 0.75 to mimic moderate and strong within subject correlation while the data-based simulation sets $\rho = 2/3$. The $\rho$ for the data-based study is close to the estimated lag 1 auto-correlation in the smoothed data. Using the multivariate normal induces a standard Brownian motion while using the multivariate $t$ produces a $t$-process. Similar to \cite{Meyer2025}, we take the value of $K$ to be large, $K = 1000$.

\begin{table} 
	\centering
	\caption{Type I Error for signed doubly ranked tests (SDRT) and the functional sign test (FST) under varying constructions: integral (Int.) and sufficient statistic (Suff.). {\bf Bolded} values indicate closest to nominal. \label{t:fdr}}
	\begin{tabular}{llllccc}
		  \hline
		  \multirow{2}{*}{$\rho$} & \multirow{2}{*}{$S$} & \multirow{2}{*}{$Z_{ijk}$} & \multirow{2}{*}{$n$} &  \multirow{2}{*}{SDRT} &  \multicolumn{2}{c}{FST}    \\
		 \cline{6-7}
		&  &  &  &  & Int. & Suff. \\ 
		  \hline
		$0.5$ & 40 & G & 15 & \bf 0.0471 & 0.0350 & 0.0274 \\
		 & & & 30 & \bf 0.0472 & 0.0452 & 0.0341 \\
		 & & & 60 & \bf 0.0508 & 0.0292 & 0.0401  \\
		\cline{4-7}
		& & T  & 15 & \bf 0.0462 & 0.0371 & 0.0282 \\
		 & & & 30 & \bf 0.0496 & 0.0439 & 0.0328 \\
		 & & & 60 & \bf 0.0474 & 0.0291 & 0.0382  \\
		 \cline{3-7}
		& 120 & G & 15 & \bf 0.0488 & 0.0366 & 0.0298 \\
		 & & & 30 & \bf 0.0492 & 0.0419 & 0.0406  \\
		 & & & 60 & \bf 0.0533 & 0.0283 & 0.0326  \\
		\cline{4-7}
		& & T  & 15 & \bf 0.0451 & 0.0360 & 0.0333  \\
		 & & & 30 & \bf 0.0508 & 0.0448 & 0.0392 \\
		 & & & 60 & \bf 0.0492 & 0.0276 & 0.0343  \\
		 \cline{3-7}
		& 360 &  G & 15 & \bf 0.0496 & 0.0334 & 0.0331 \\
		 & & & 30 & \bf 0.0495 & 0.0426 & 0.0439 \\
		 & & & 60 & \bf 0.0503 & 0.0282 & 0.0302  \\
		\cline{4-7}
		& & T  & 15 & \bf 0.0473 & 0.0335 & 0.0324 \\
		 & & & 30 & \bf 0.0510 & 0.0432 & 0.0423 \\
		 & & & 60  & \bf 0.0475 & 0.0279 & 0.0321  \\
		\cline{2-7}
		$0.75$ & 40 & G & 15 & \bf 0.0481 & 0.0337 & 0.0234  \\
		 & & & 30 & \bf 0.0485 & 0.0453 & 0.0317  \\
		 & & & 60  & \bf 0.0507 & 0.0299 & 0.0381 \\
		\cline{4-7}
		& & T  & 15 & \bf 0.0467 & 0.0363 & 0.0278 \\
		 & & & 30 & \bf 0.0504 & 0.0441 & 0.0339 \\
		 & & & 60 & \bf 0.0474 & 0.0259 & 0.0402  \\
		 \cline{3-7}
		& 120 & G & 15 & \bf 0.0490 & 0.0365 & 0.0329 \\
		 & & & 30 & \bf 0.0490 & 0.0423 & 0.0387 \\
		 & & & 60 & \bf 0.0530 & 0.0309 & 0.0356  \\
		\cline{4-7}
		& & T  & 15 & \bf 0.0455 & 0.0375 & 0.0328 \\
		 & & & 30 &  0.0492 & \bf 0.0505 & 0.0405 \\
		 & & & 60 & \bf 0.0510 & 0.0285 & 0.0338  \\
		 \cline{3-7}
		& 360 &  G & 15 & \bf 0.0475 & 0.0359 & 0.0334 \\
		 & & & 30 & \bf 0.0492 & 0.0445 & 0.0436 \\
		 & & & 60 & \bf 0.0498 & 0.0265 & 0.0307  \\
		\cline{4-7}
		& & T  & 15 & \bf 0.0475 & 0.0349 & 0.0352 \\
		 & & & 30 & \bf 0.0506 & 0.0446 & 0.0427 \\
		 & & & 60  & \bf 0.0467 & 0.0262 & 0.0297 \\
		 \hline
	\end{tabular}
\end{table}
		
	The final step in generating the data involves adding measurement error and a potentially non-zero shift. Thus, we generate pairs of functional data from $Y_{i0}(s) = X_{i0}(s) + \epsilon_{i0}(s)$ and $Y_{i1}(s) = \Delta(s) + X_{i1}(s) + \epsilon_{i1}(s)$ where $\Delta(s)$ is either set to zero to assess type I error or has one of two functional forms:
%	\begin{align*}
		$\Delta_1(s) = \xi s$ and $\Delta_2(s) = \xi4s(1-s)$.
%		, \text{ and}\\
%		\Delta_3(s) = \xi m^{-1} &B(2, 6)^{-1} s^{2-1} (1-s)^{6-1}.
%	\end{align*}
	To evaluate power, we let $\xi$ take on a value between 0.12 and 3, incrementing by 0.12. These functions are similar to those evaluated by \cite{Chak2015} and \cite{Meyer2025} in the independent, two sample case. The error terms, $\epsilon_{i0}(s)$ and $\epsilon_{i1}(s)$, mimic potential error at each measurement occurrence and come from an AR(1) process with correlation set to 0.5 and variance set to 1. For the general study, we preprocess $Y_{i0}(s)$ and $Y_{i1}(s)$ using FACE and retain 99\% of the variability. Because the data contains some missing values, for the data-based study, we preprocess using FPCA SC---once again retaining 99\% of the variability. 

\begin{table}
	\centering
	\caption{Type I Error when $\rho = 2/3$ and $S = 80$ for the data-based simulation using with FPCA SC. {\bf Bolded} values indicate the closest value to nominal, $\alpha = 0.05$. SDRT denotes signed doubly ranked test, FST denotes functional sign test, Int. denotes integral, and Suff. denotes sufficient statistic.\label{t:fdrdb}}
	\begin{tabular}{llccc}
		  \hline
		  \multirow{2}{*}{$Z_{ijk}$} & \multirow{2}{*}{$n$} & \multirow{2}{*}{SDRT} &  \multicolumn{2}{c}{FST}    \\
		 \cline{4-5}
		 &  &  & Int. & Suff. \\ 
		  \hline
		 G & 15 & \bf 0.0490 & 0.0361 & 0.0284  \\
		  & 30 & \bf 0.0489 & 0.0418 & 0.0334  \\
		  & 60  & \bf 0.0518 & 0.0264 & 0.0367 \\
		\cline{2-5}
		T  & 15 & \bf 0.0462 & 0.0401 & 0.0321  \\
		 & 30 & \bf 0.0494 & 0.0445 & 0.0395  \\
		 & 60  & \bf 0.0512 & 0.0283 & 0.0357 \\
		 \hline
	\end{tabular}
\end{table}

\begin{figure}
	\centering
	\includegraphics[width =2.1in, height =2.1in]{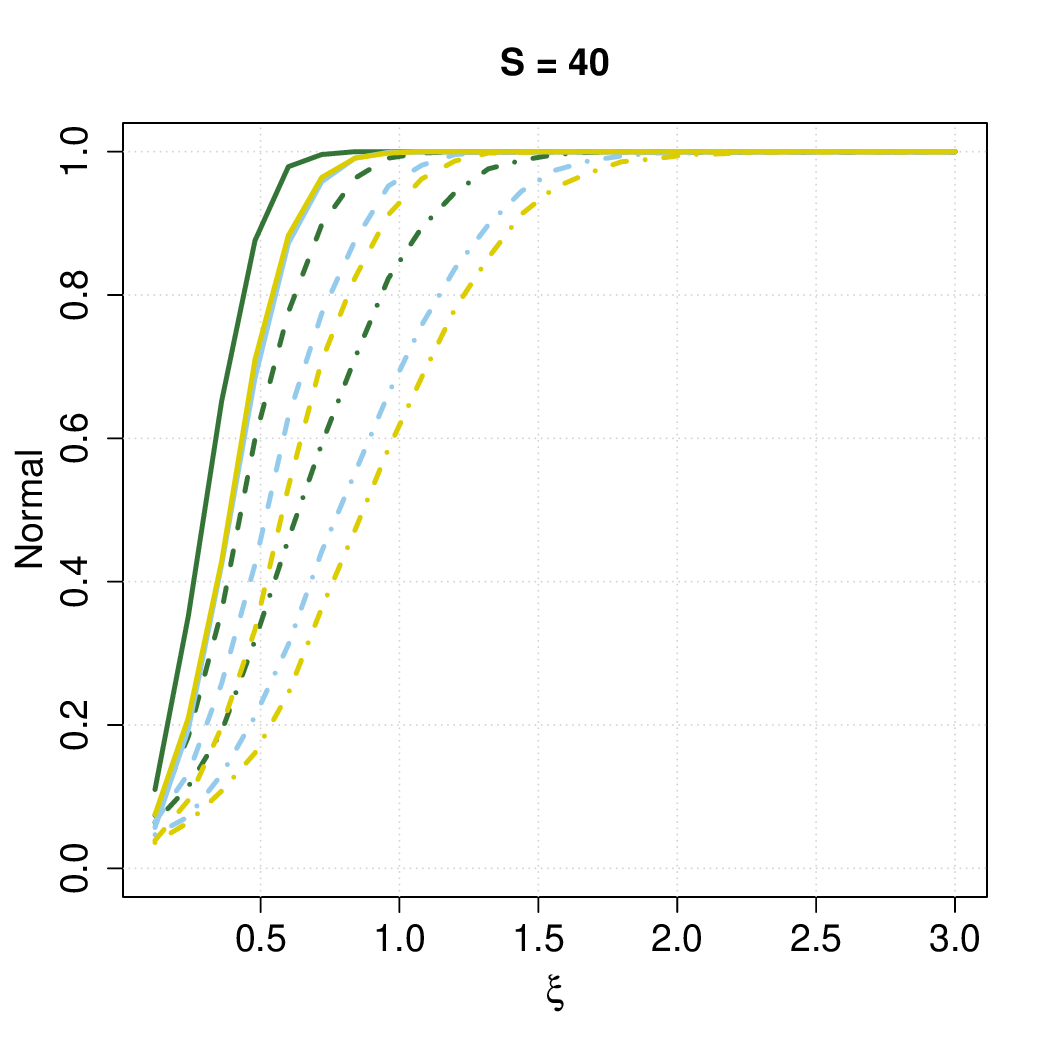}
	\includegraphics[width =2.1in, height =2.1in]{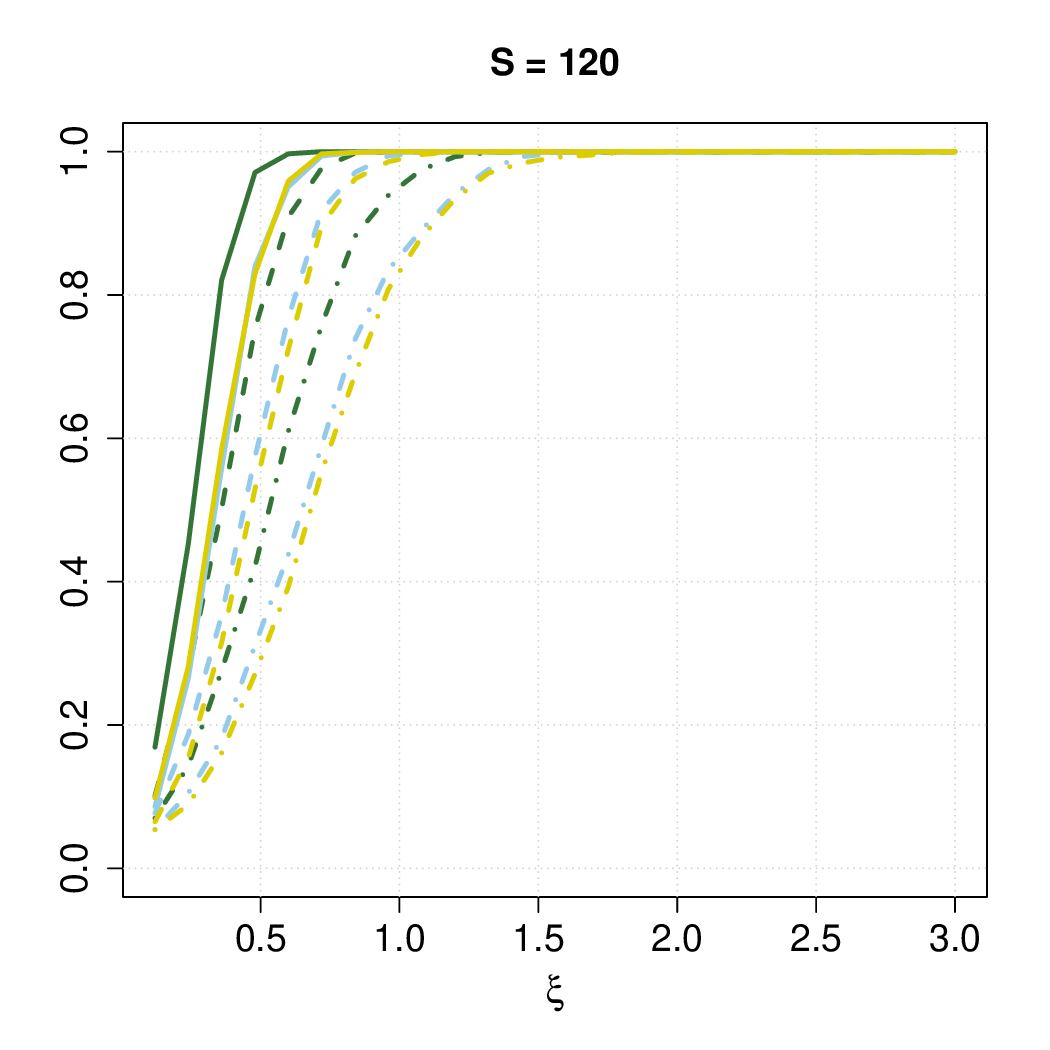}
	\includegraphics[width =2.1in, height =2.1in]{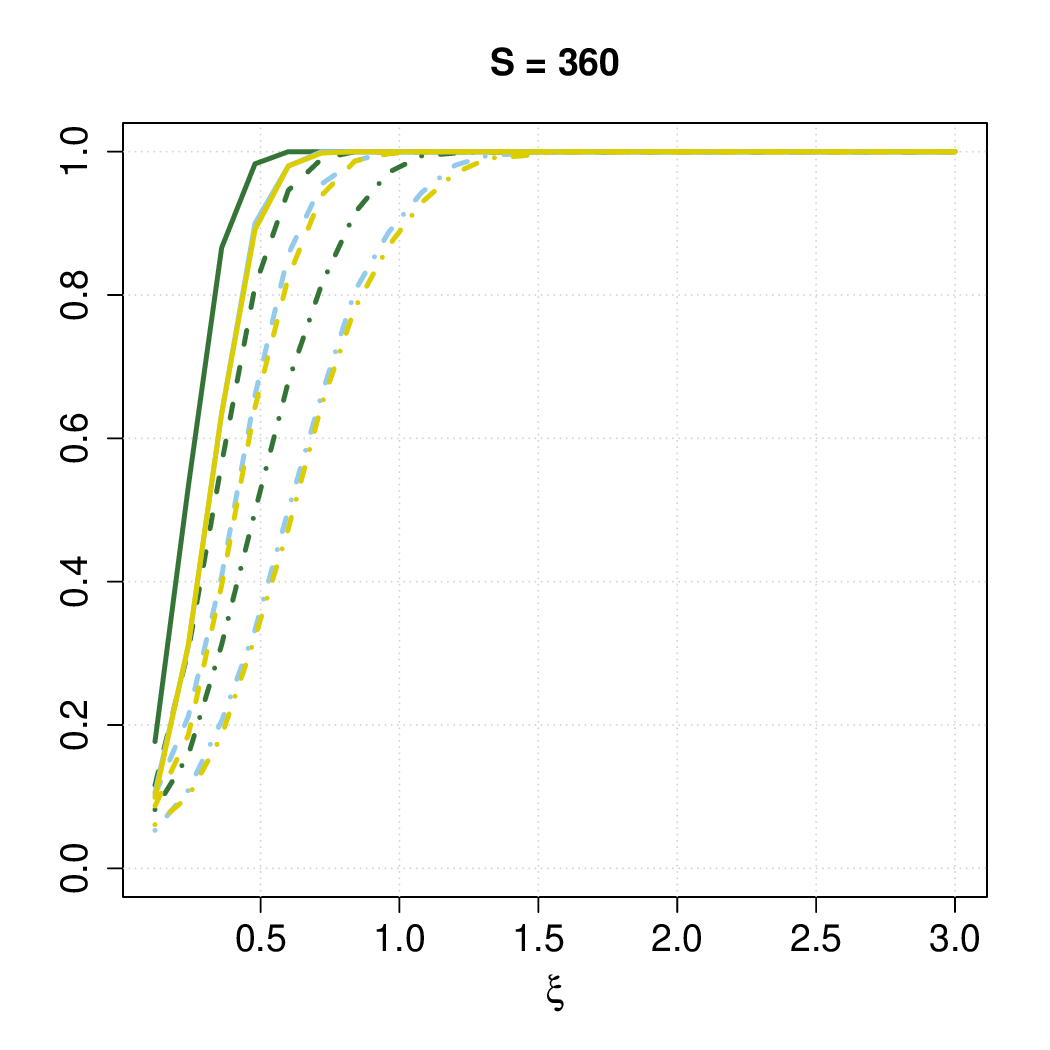}
	\includegraphics[width =2.1in, height =2.1in]{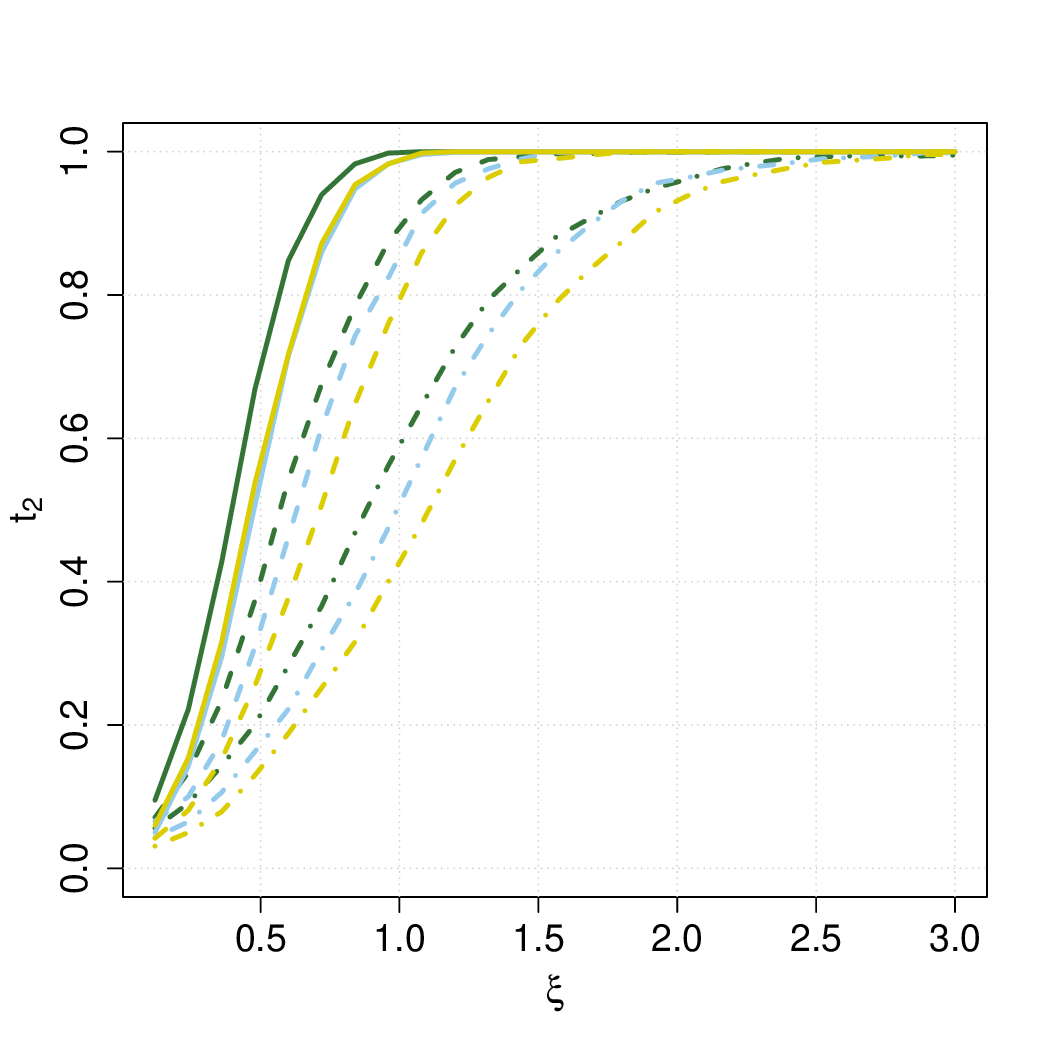}
	\includegraphics[width =2.1in, height =2.1in]{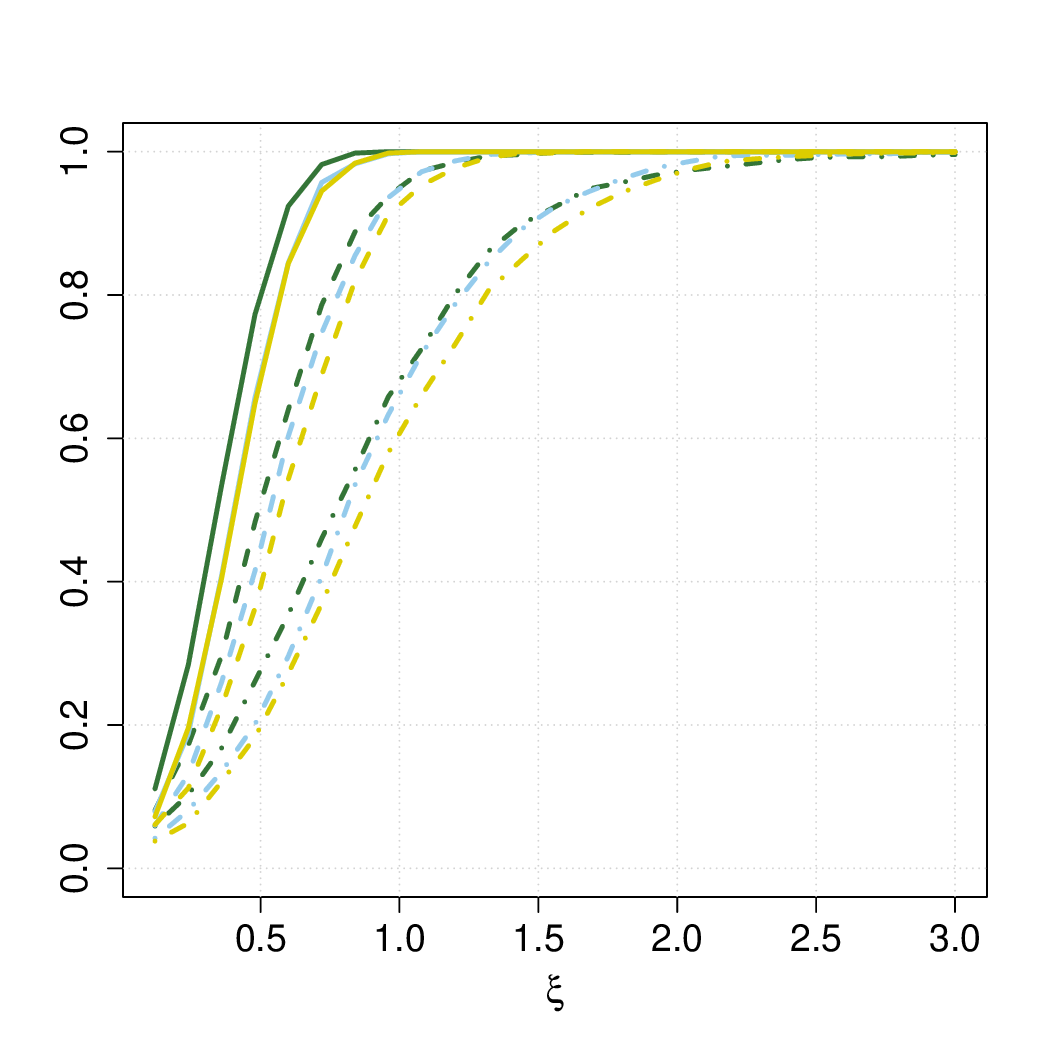}
	\includegraphics[width =2.1in, height =2.1in]{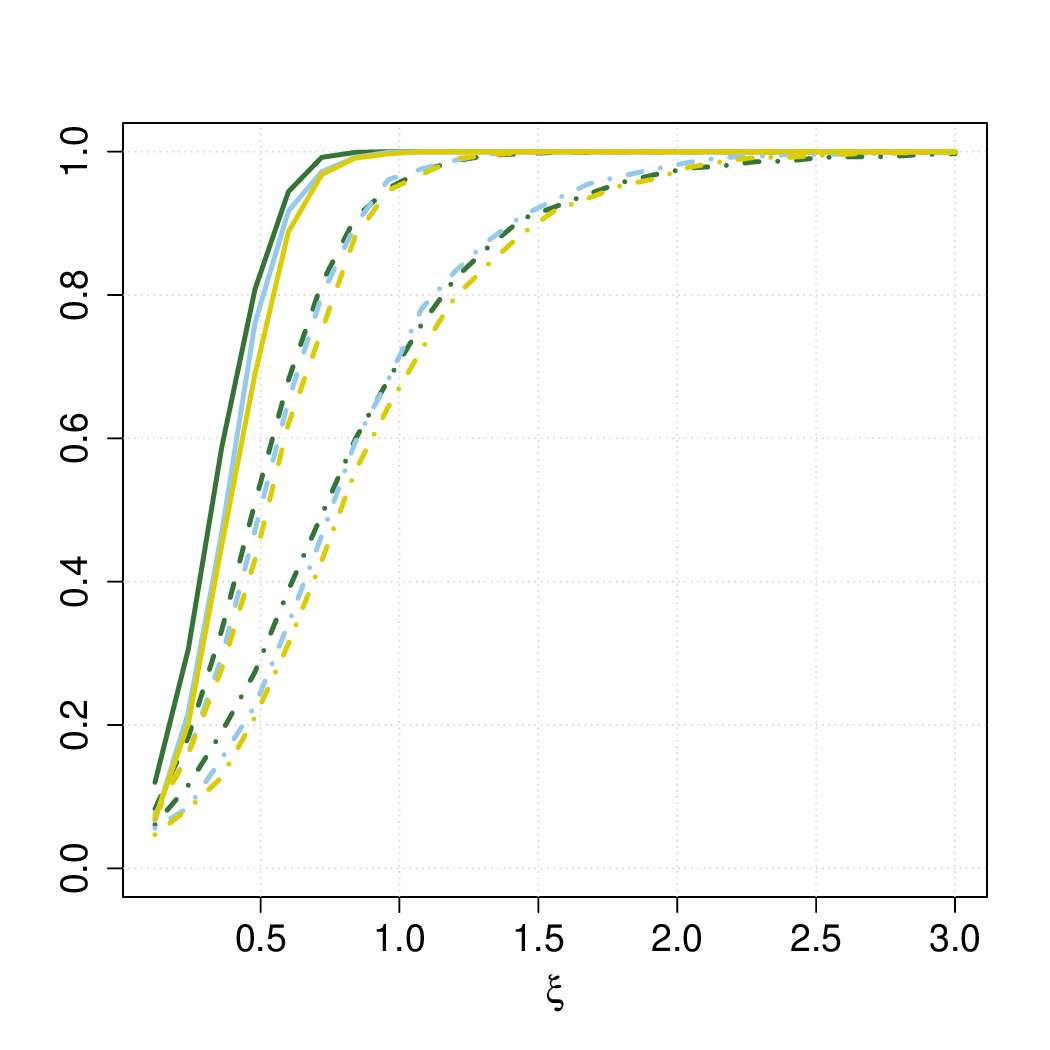}
	\caption{Power curves for functional sign and signed doubly ranked tests under $\Delta_1(s)$ and $\rho = 0.75$. The top row contains results under the multivariate normal, bottom row under $t_2$. Curves for the signed doubly ranked tests are in green, while curves for the functional sign tests are in blue (integral) and gold (sufficient statistic). Solid curves are for when $n = 60$, dashed curves for when $n = 30$, and dotted-dashed curves for when $n = 15$.\label{f:pone}}
\end{figure}

\begin{figure}
	\centering
	\includegraphics[width =2.1in, height =2.1in]{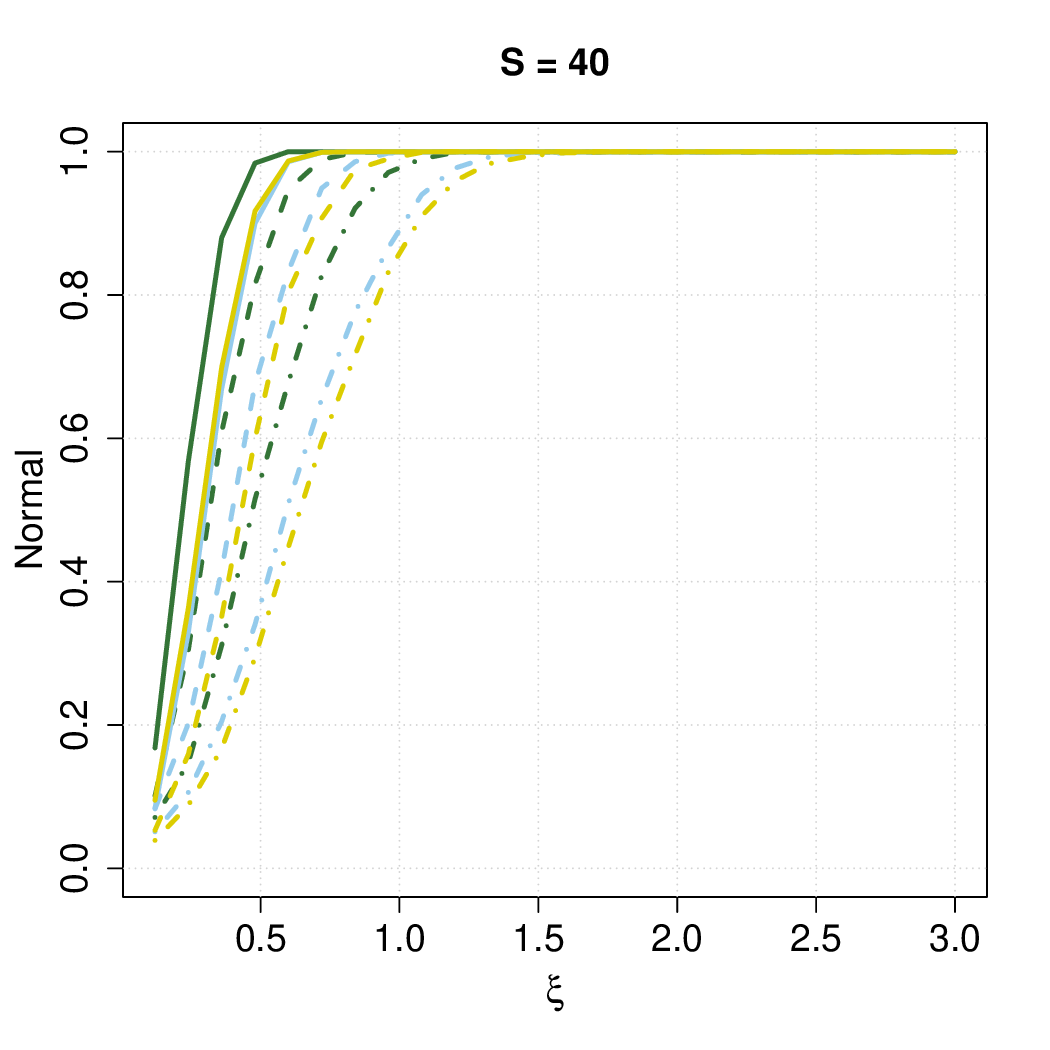}
	\includegraphics[width =2.1in, height =2.1in]{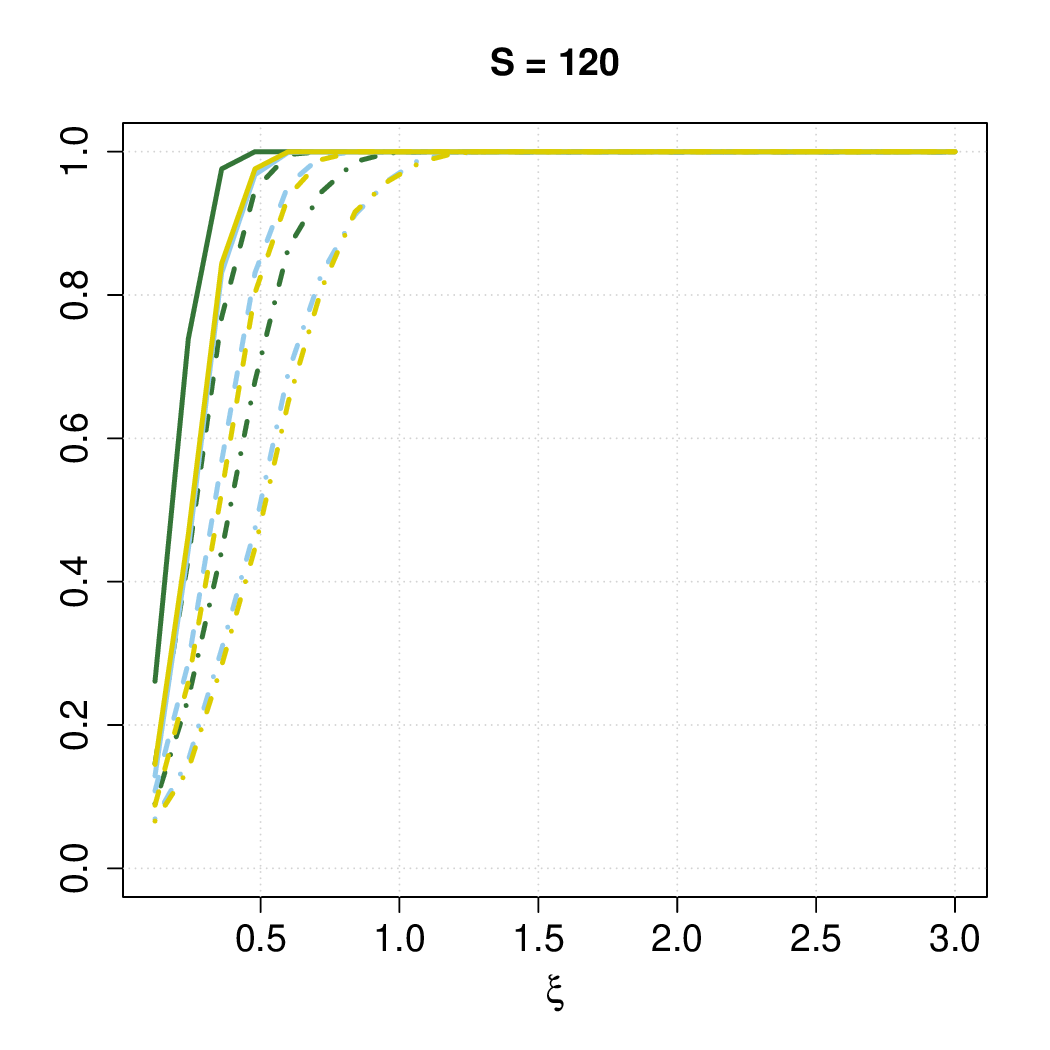}
	\includegraphics[width =2.1in, height =2.1in]{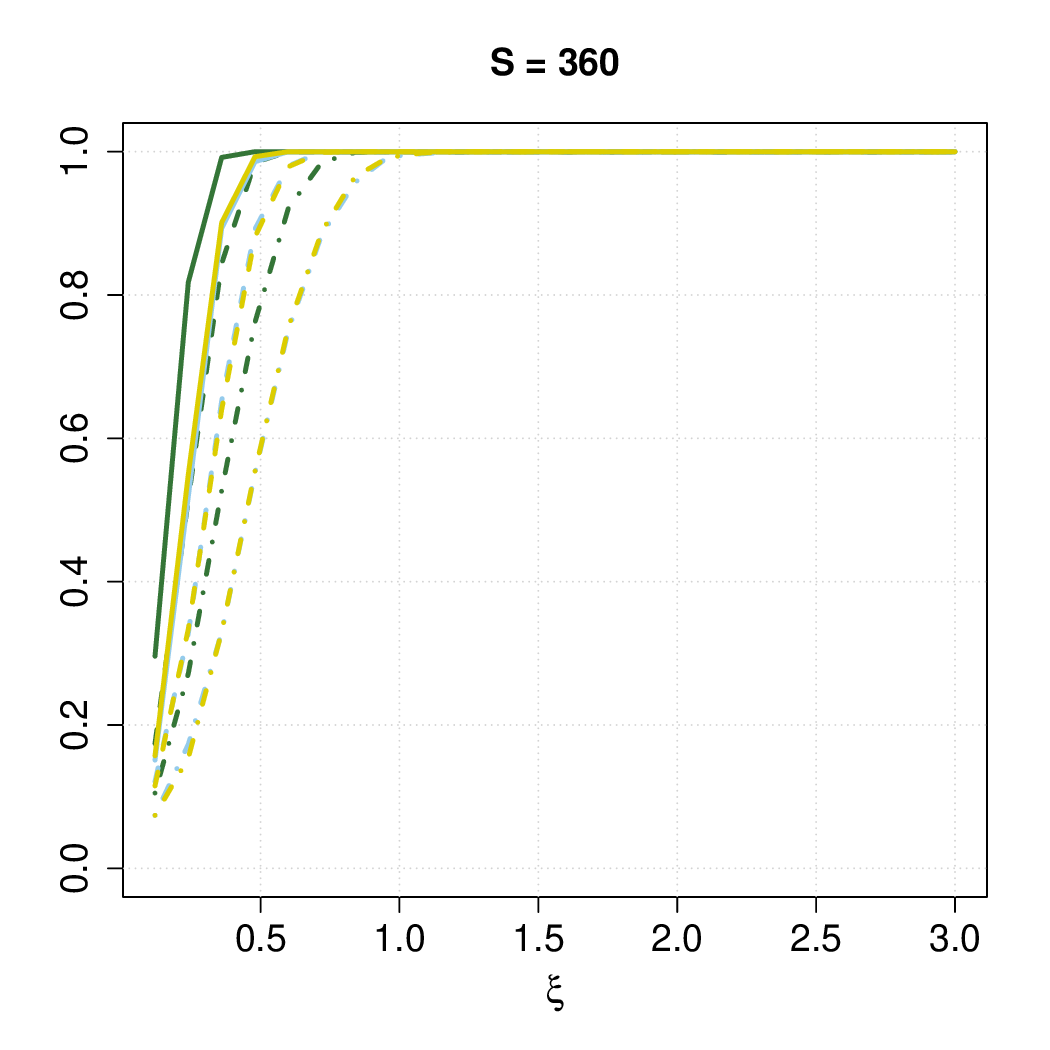}
	\includegraphics[width =2.1in, height =2.1in]{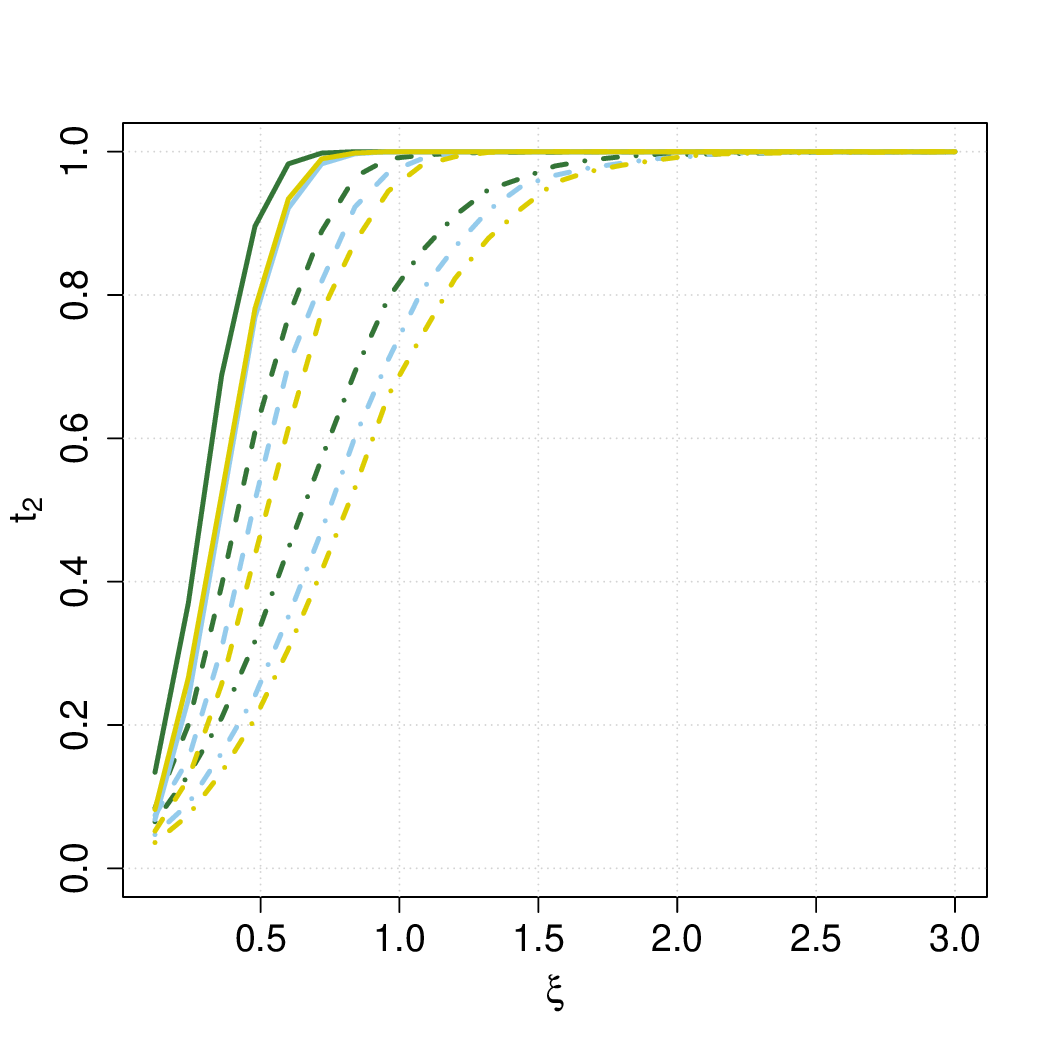}
	\includegraphics[width =2.1in, height =2.1in]{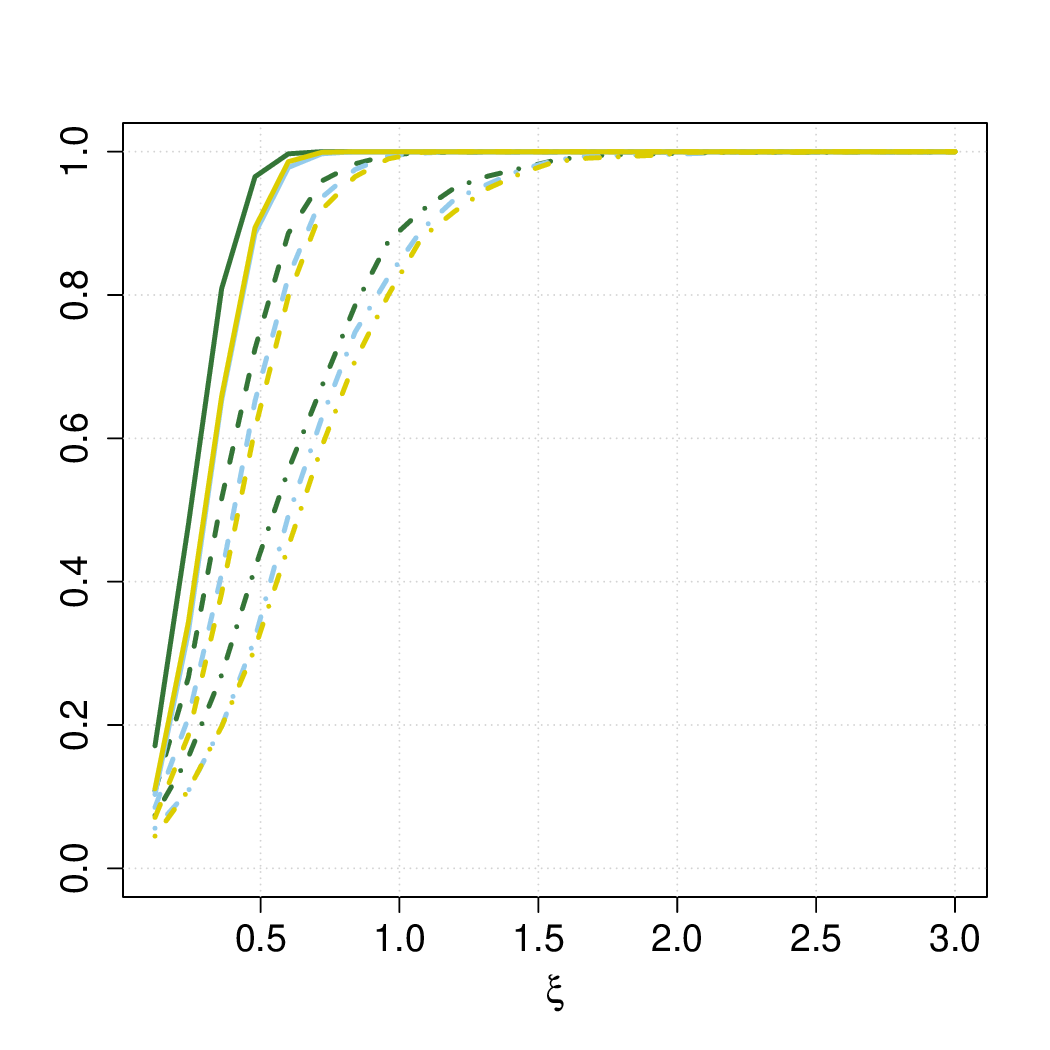}
	\includegraphics[width =2.1in, height =2.1in]{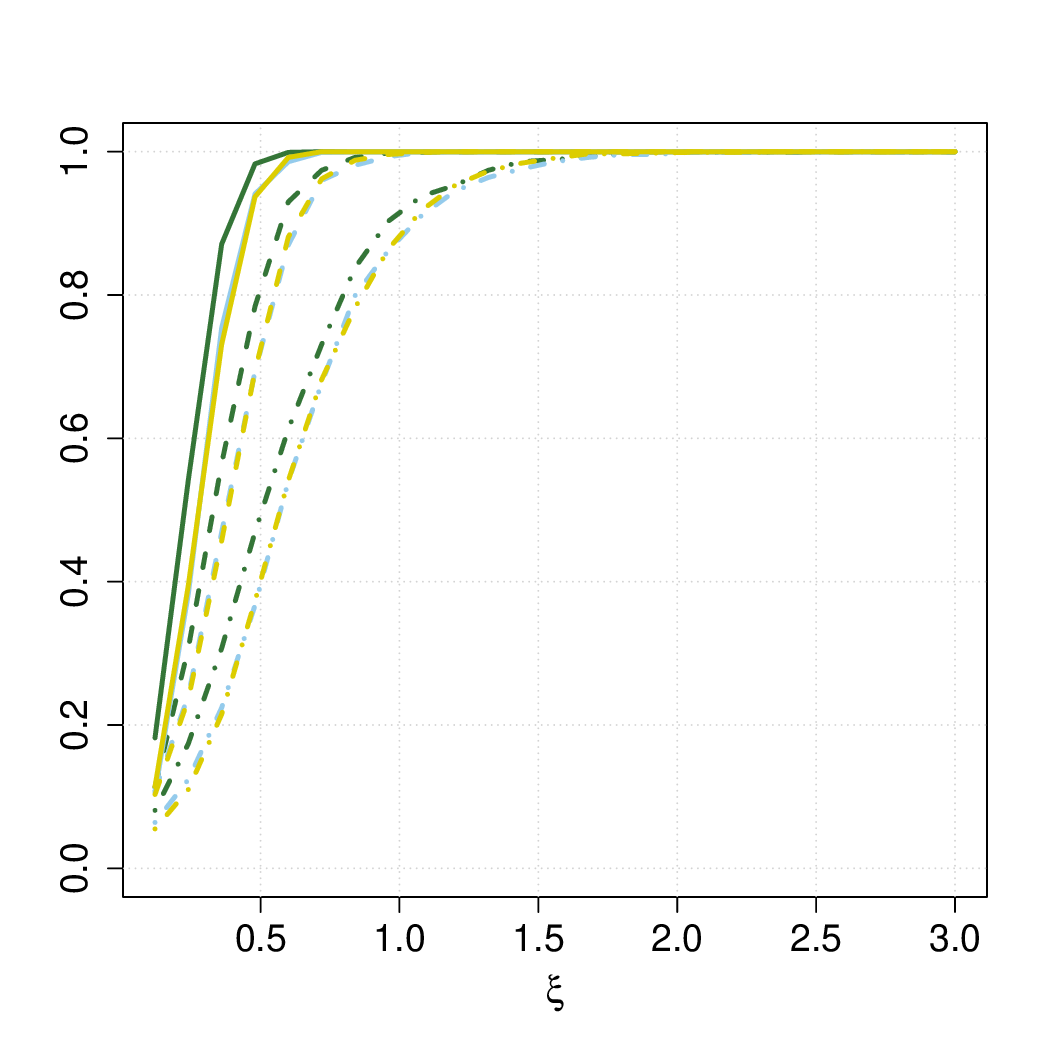}
	\caption{Power curves for functional sign and signed doubly ranked tests under $\Delta_2(s)$ and $\rho = 0.75$. The top row contains results under the multivariate normal, bottom row under $t_2$. Curves for the signed doubly ranked tests are in green, while curves for the functional sign tests are in blue (integral) and gold (sufficient statistic). Solid curves are for when $n = 60$, dashed curves for when $n = 30$, and dotted-dashed curves for when $n = 15$.\label{f:ptwo}}
\end{figure}

\begin{figure}
	\centering
	\includegraphics[width =2.1in, height =2.1in]{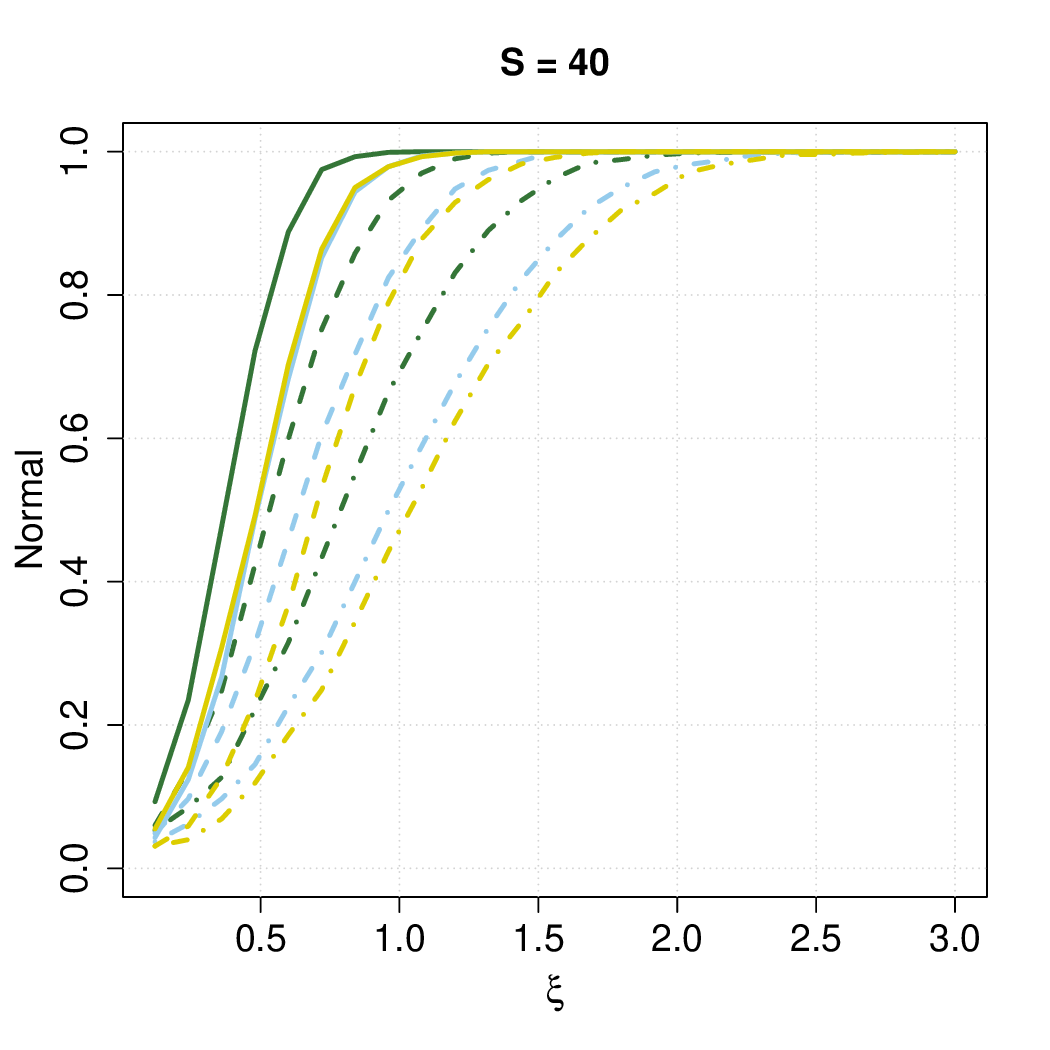}
	\includegraphics[width =2.1in, height =2.1in]{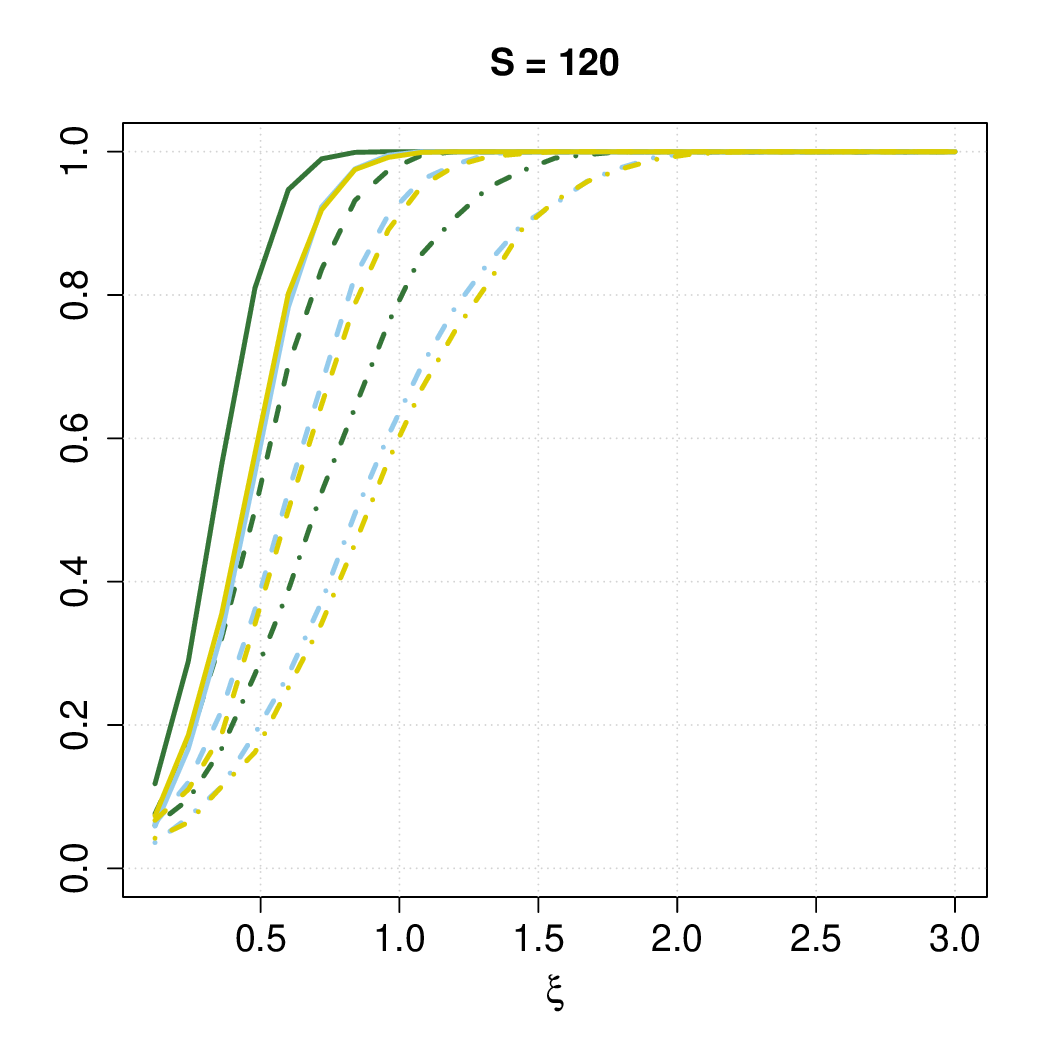}
	\includegraphics[width =2.1in, height =2.1in]{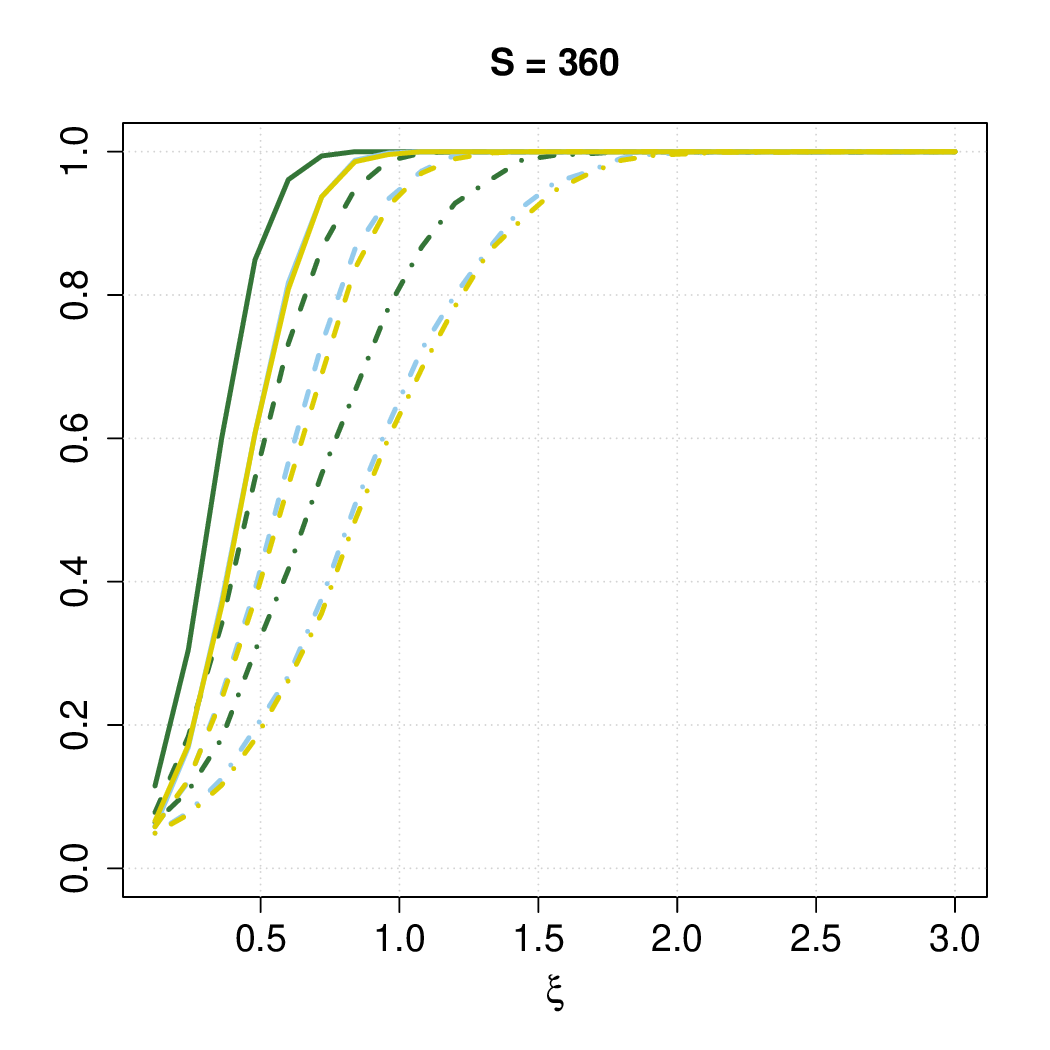}
	\includegraphics[width =2.1in, height =2.1in]{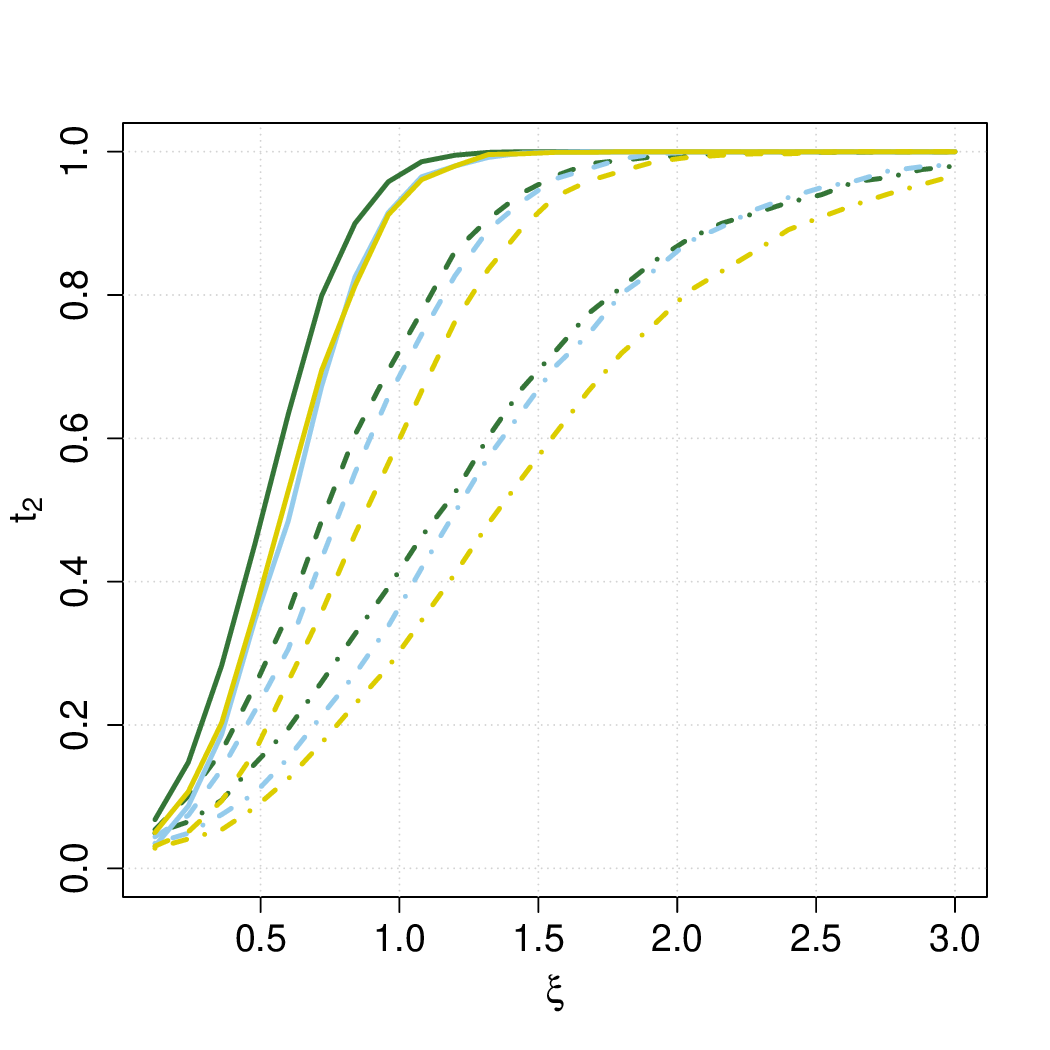}
	\includegraphics[width =2.1in, height =2.1in]{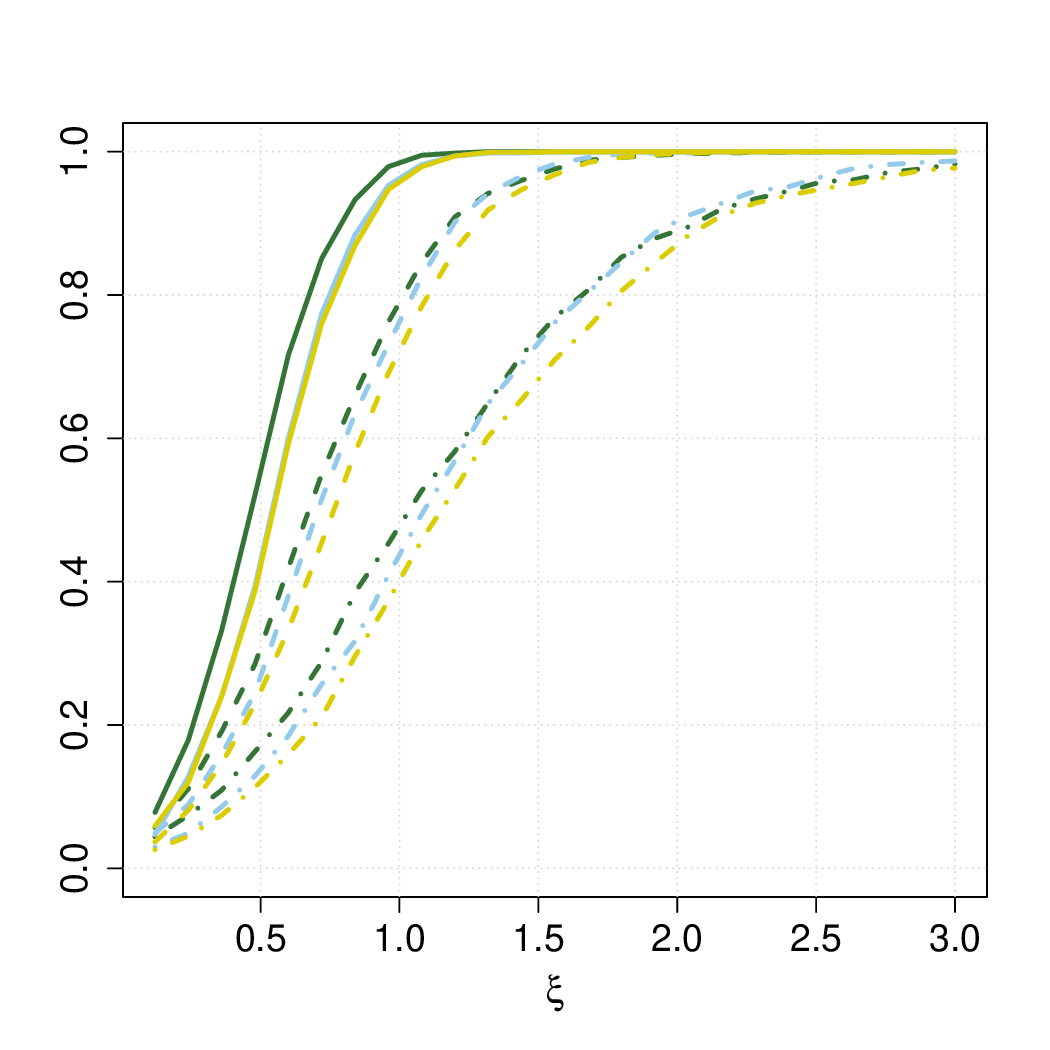}
	\includegraphics[width =2.1in, height =2.1in]{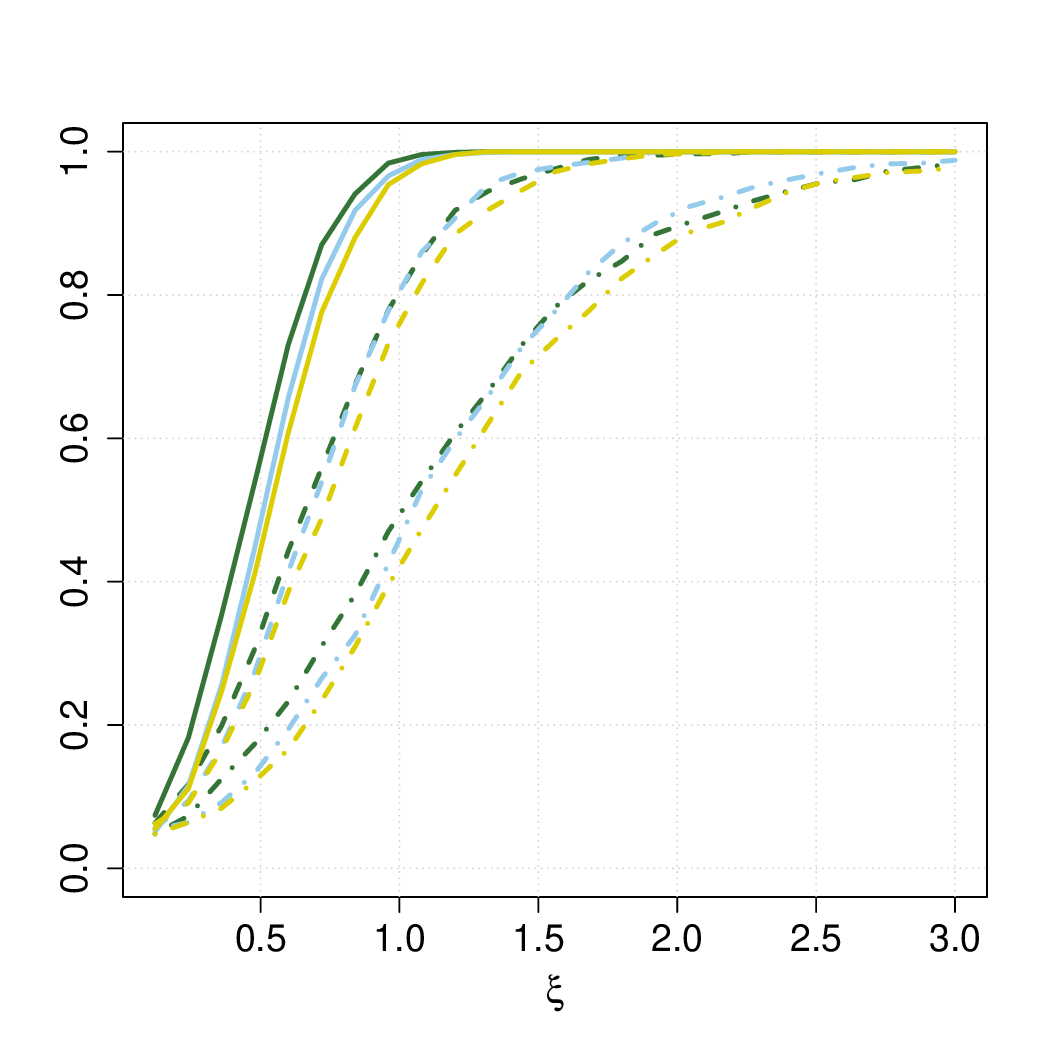}
	\caption{Power curves for functional sign and signed doubly ranked tests under $\Delta_1(s)$ and $\rho = 0.50$. The top row contains results under the multivariate normal, bottom row under $t_2$. Curves for the signed doubly ranked tests are in green, while curves for the functional sign tests are in blue (integral) and gold (sufficient statistic). Solid curves are for when $n = 60$, dashed curves for when $n = 30$, and dotted-dashed curves for when $n = 15$.\label{f:pone50}}
\end{figure}

\begin{figure}
	\centering
	\includegraphics[width =2.1in, height =2.1in]{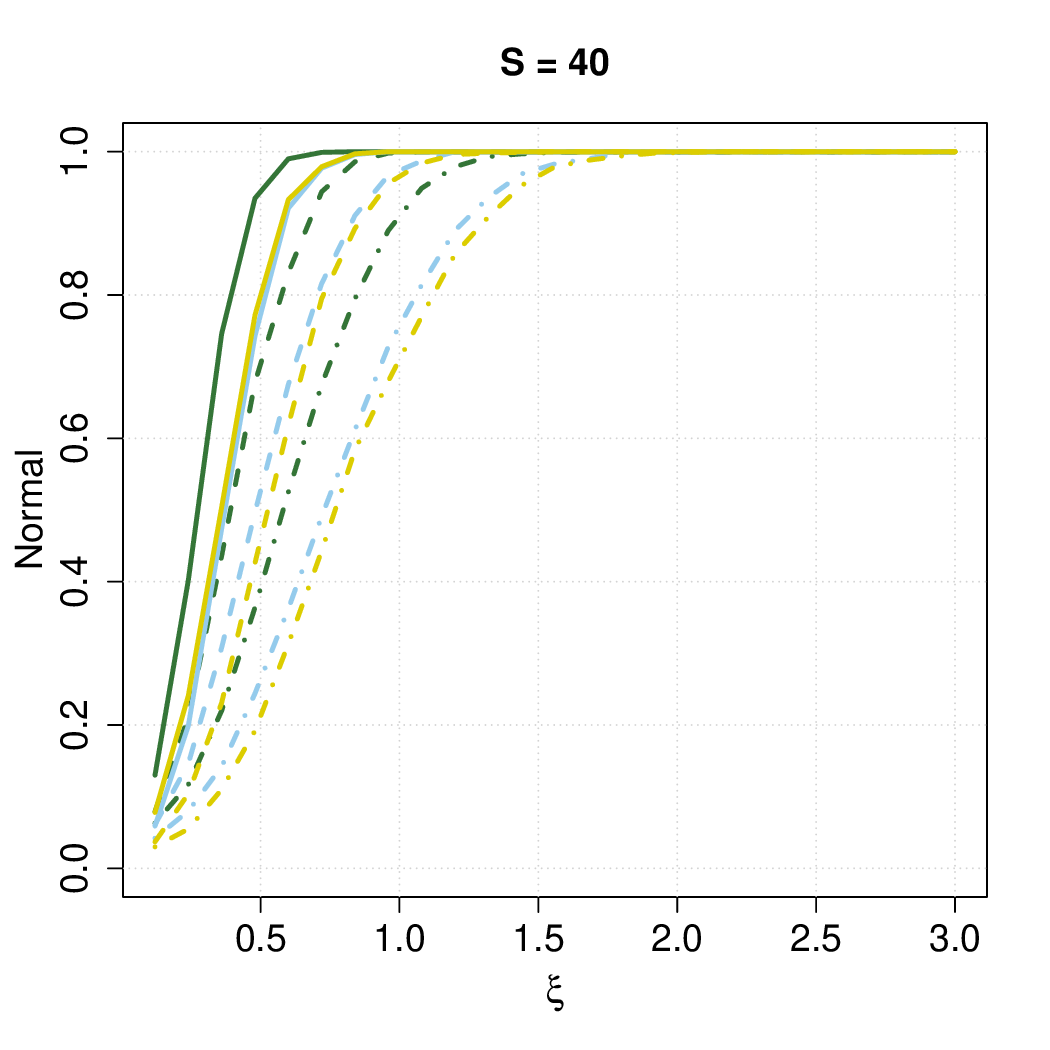}
	\includegraphics[width =2.1in, height =2.1in]{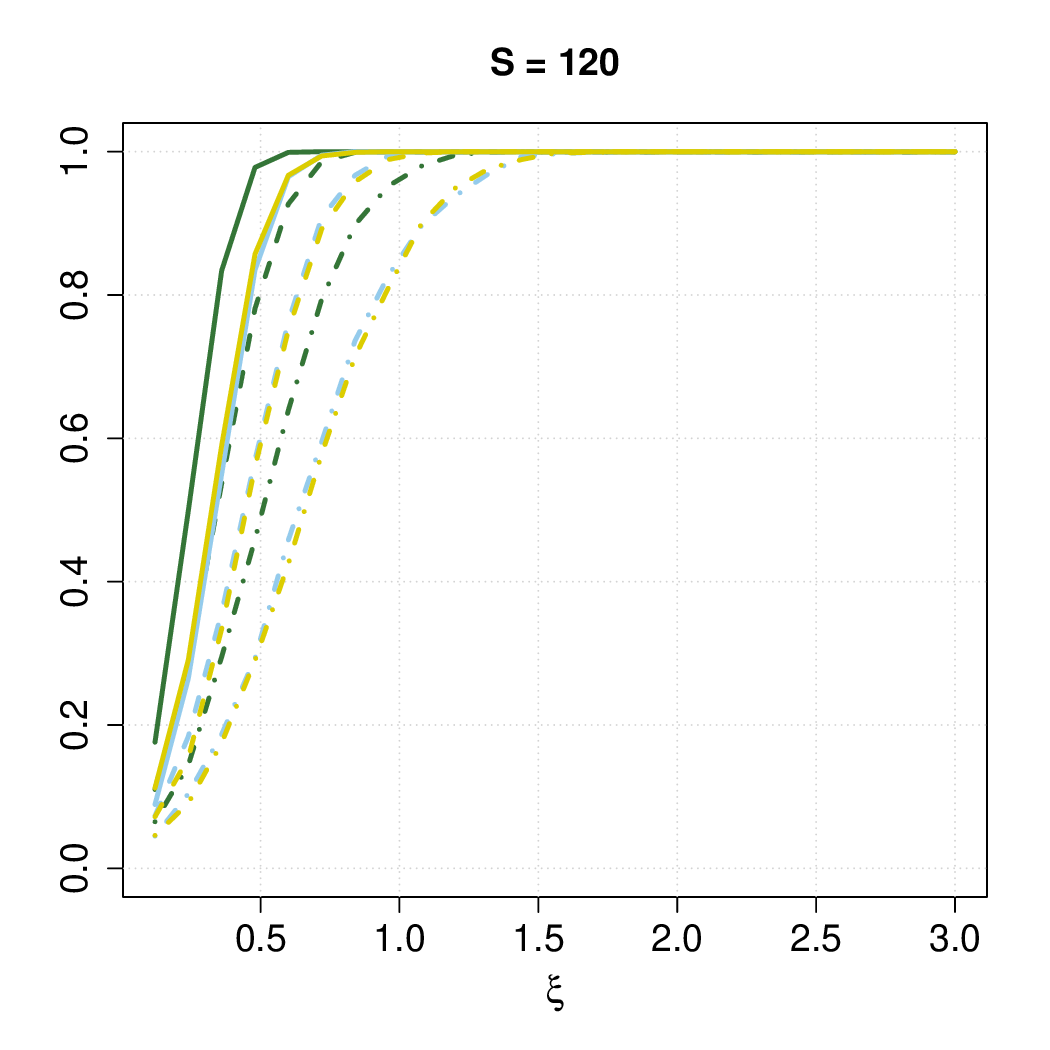}
	\includegraphics[width =2.1in, height =2.1in]{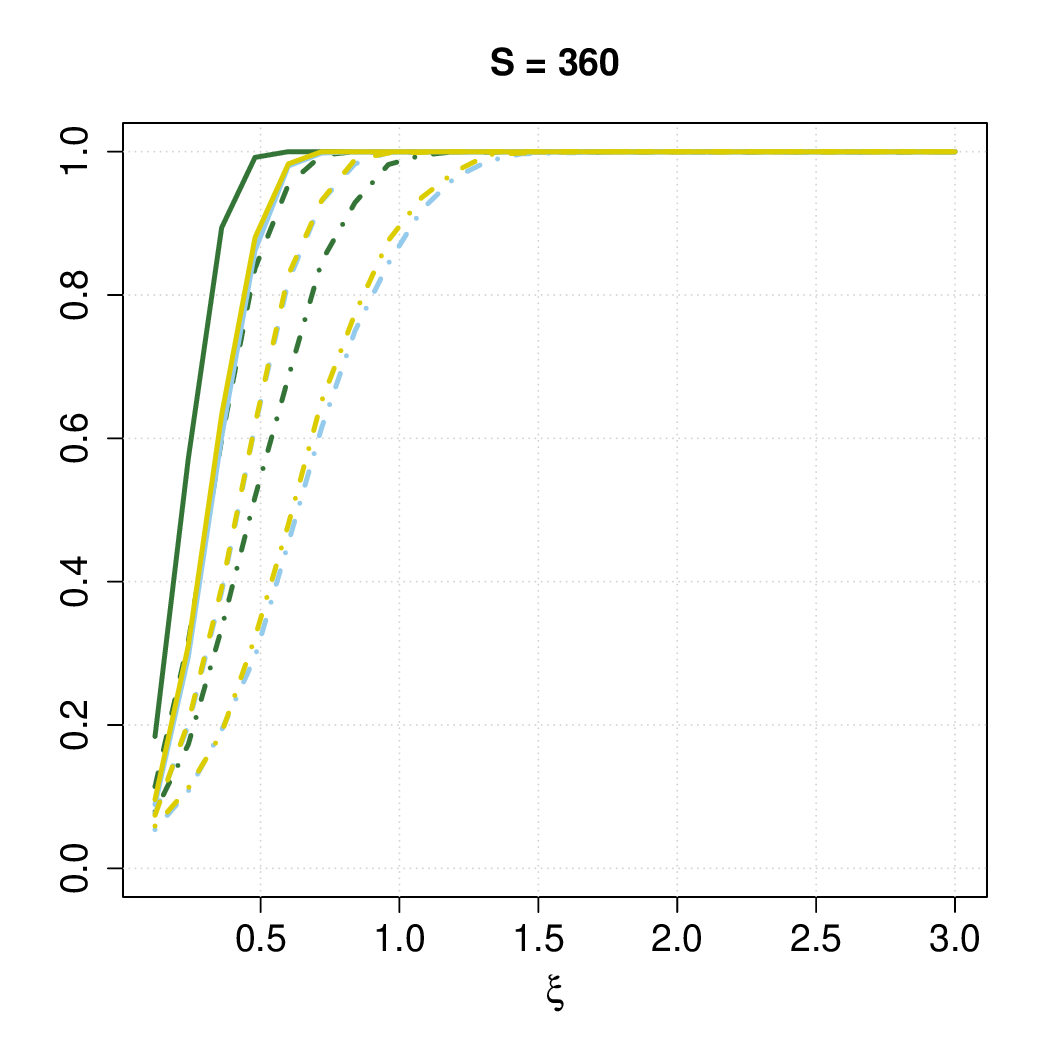}
	\includegraphics[width =2.1in, height =2.1in]{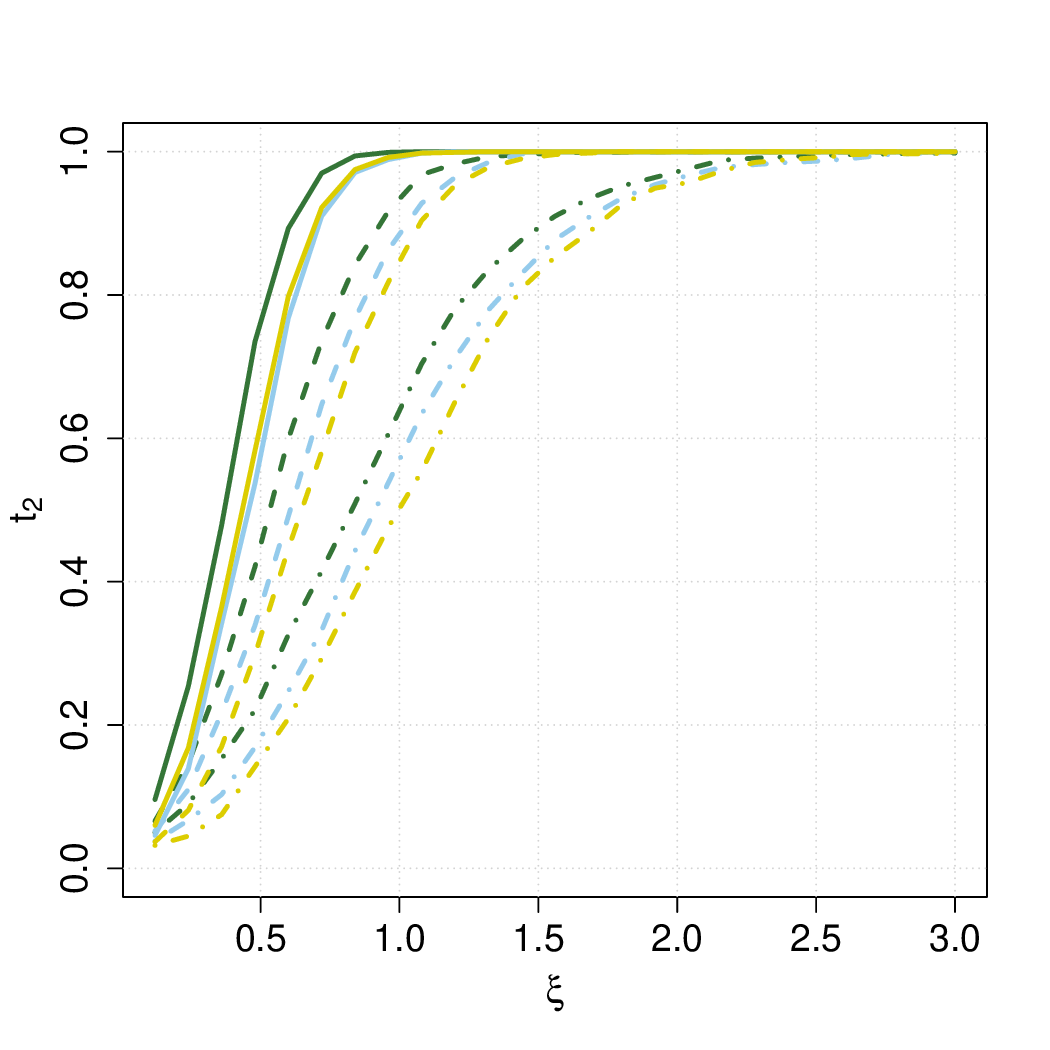}
	\includegraphics[width =2.1in, height =2.1in]{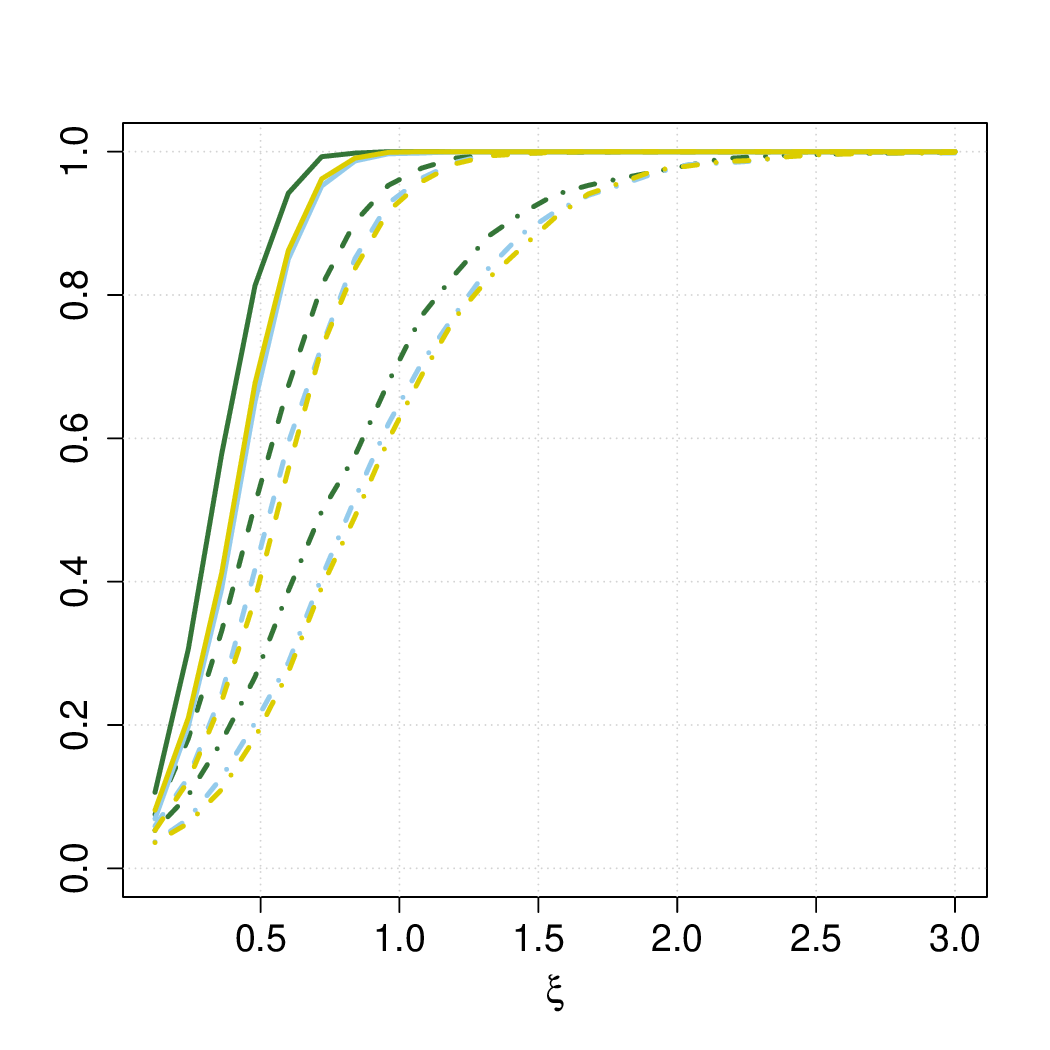}
	\includegraphics[width =2.1in, height =2.1in]{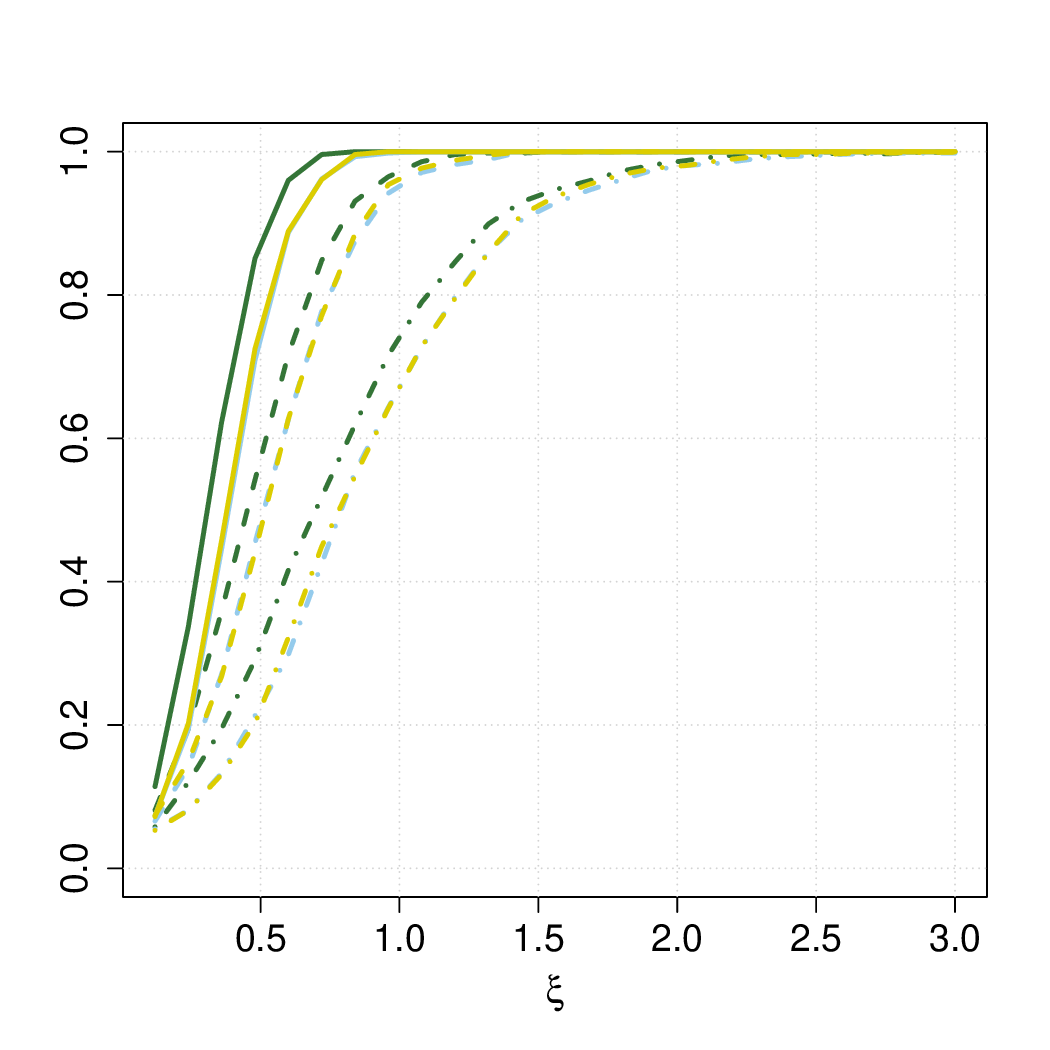}
	\caption{Power curves for functional sign and signed doubly ranked tests under $\Delta_2(s)$ and $\rho = 0.50$. The top row contains results under the multivariate normal, bottom row under $t_2$. Curves for the signed doubly ranked tests are in green, while curves for the functional sign tests are in blue (integral) and gold (sufficient statistic). Solid curves are for when $n = 60$, dashed curves for when $n = 30$, and dotted-dashed curves for when $n = 15$.\label{f:ptwo50}}
\end{figure}

\begin{figure}
	\centering
	\includegraphics[width = 2.42in, height = 2.42in]{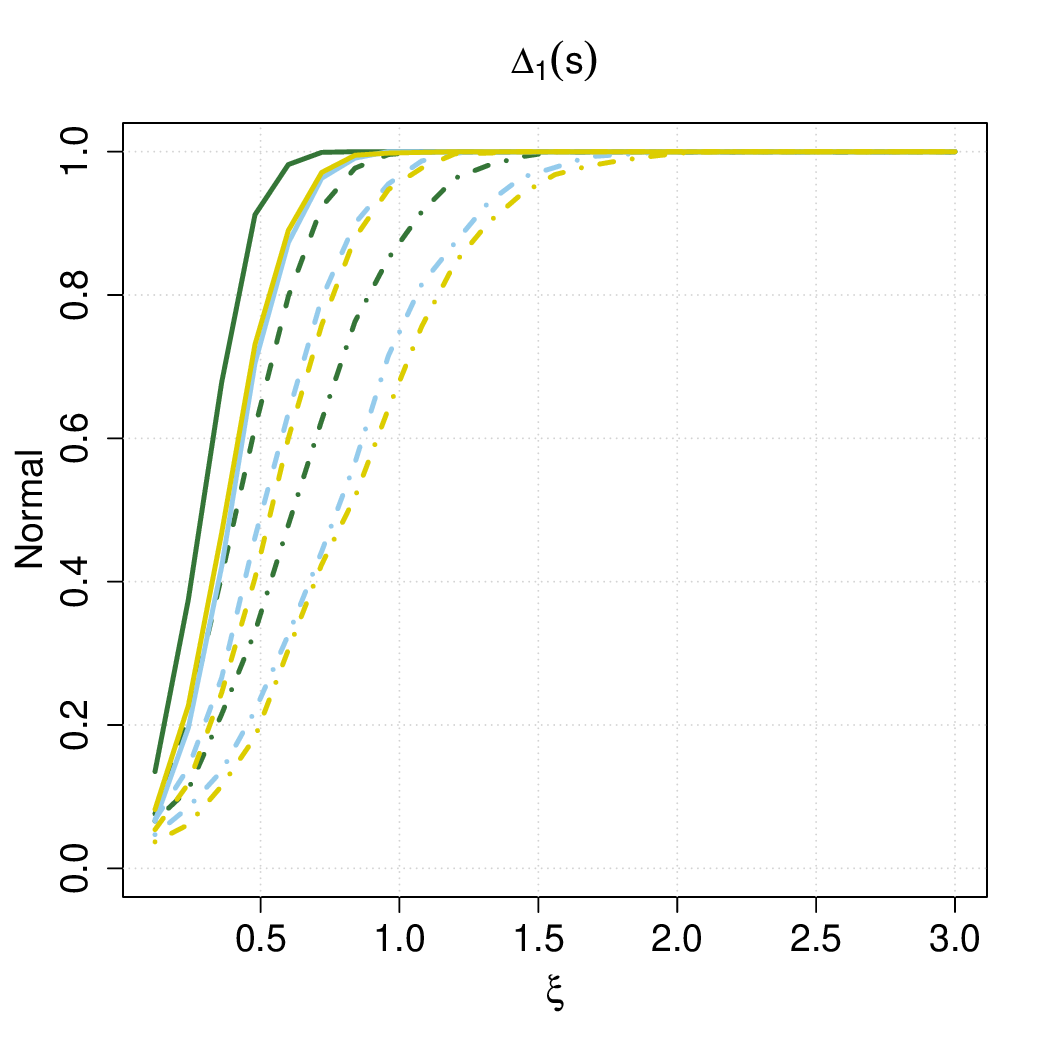}
	\includegraphics[width = 2.42in, height = 2.42in]{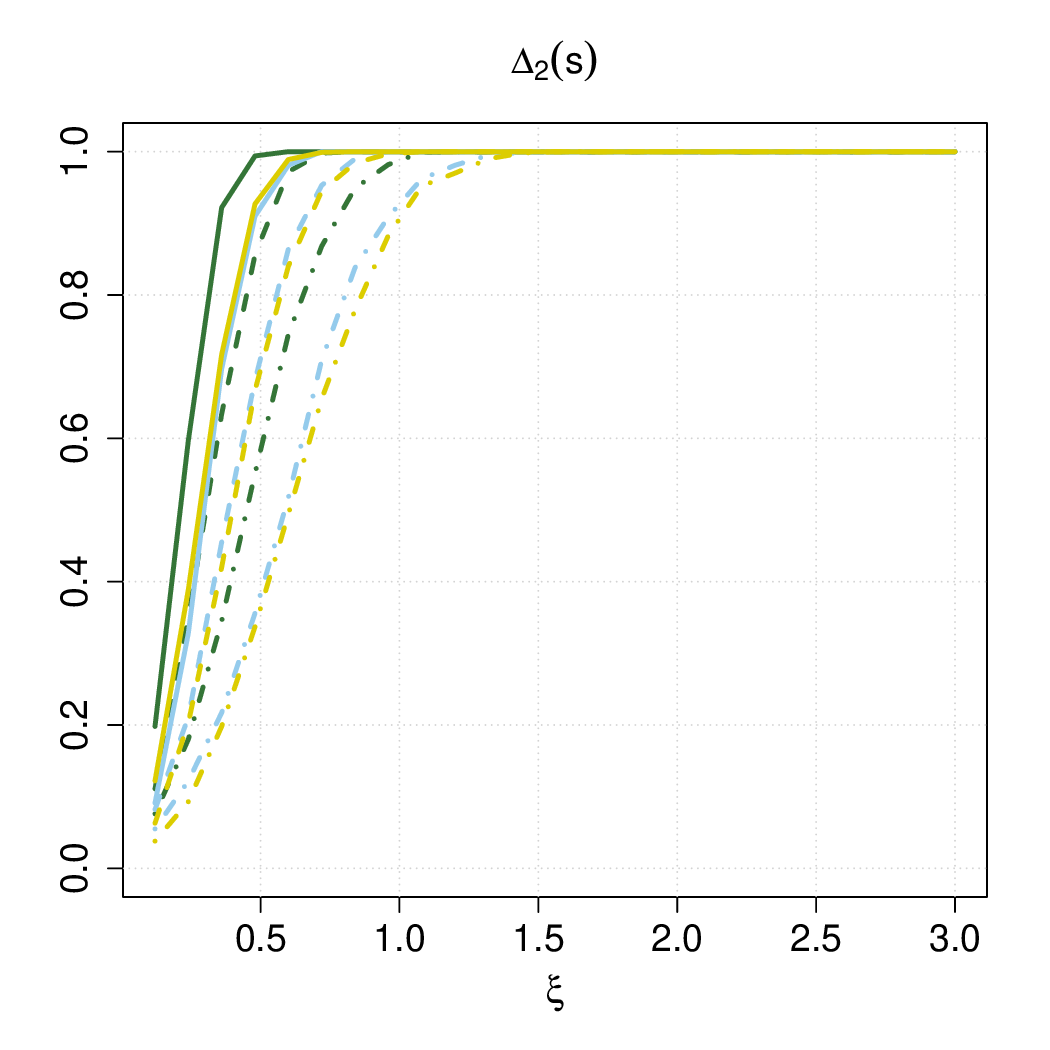}
	\includegraphics[width = 2.42in, height = 2.42in]{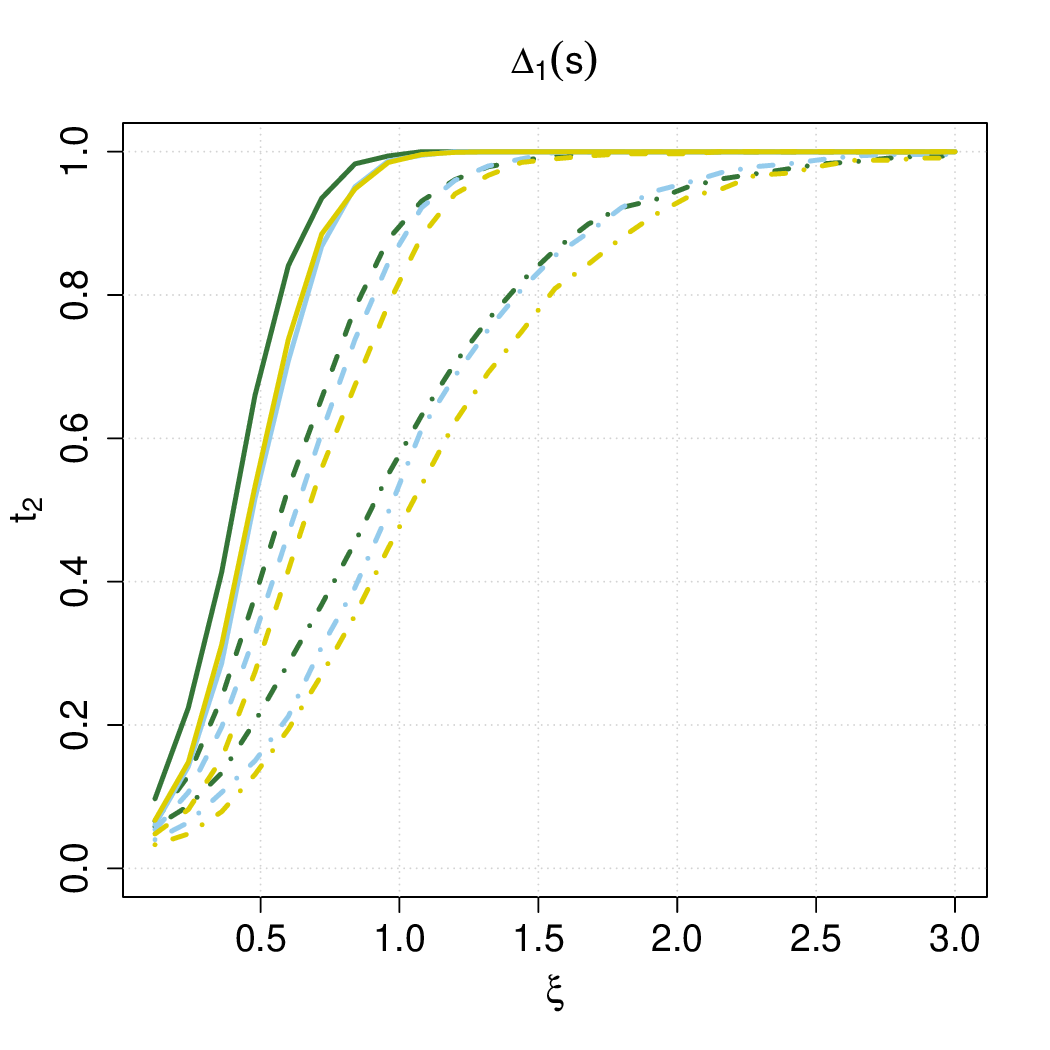}
	\includegraphics[width = 2.42in, height = 2.42in]{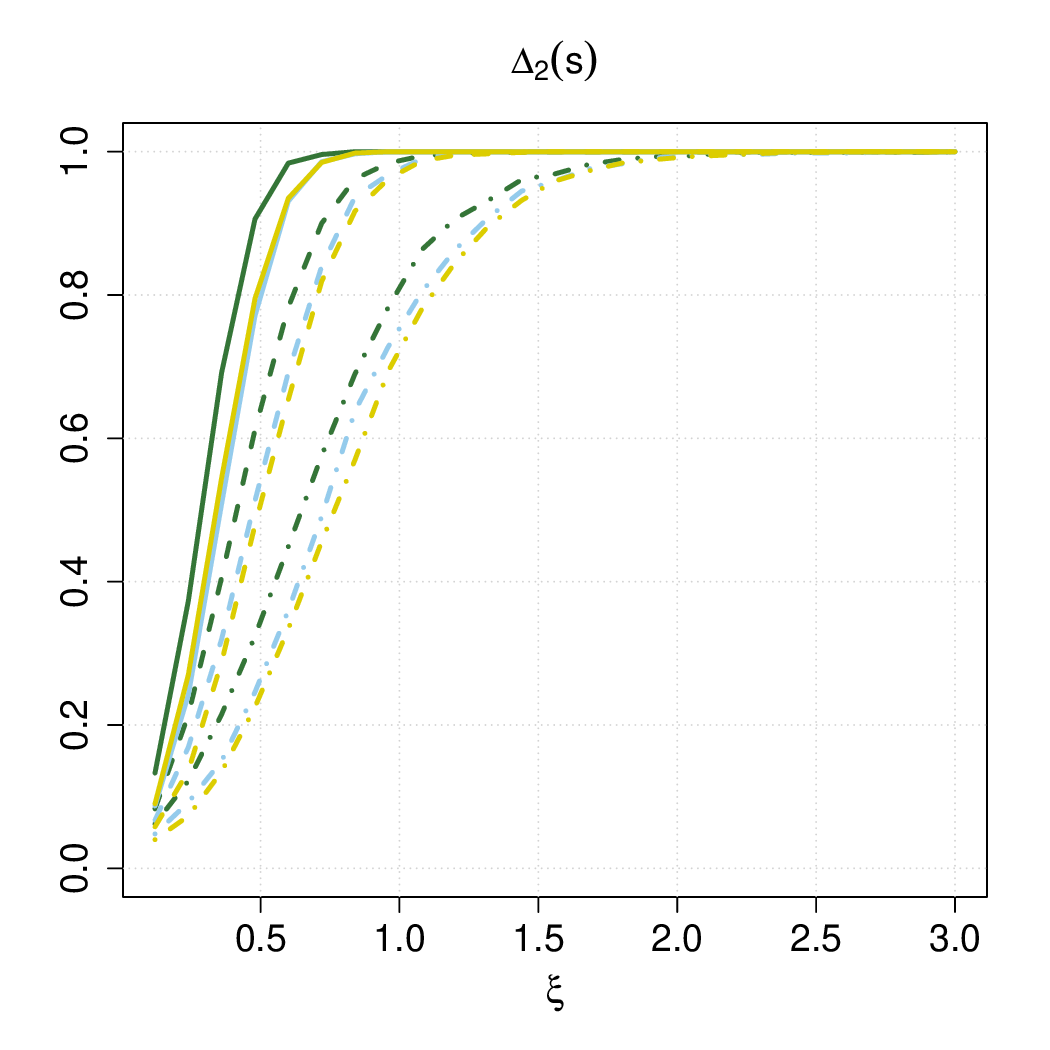}
	\caption{Power curves for functional sign and signed doubly ranked tests under the data-based setting. The top row contains results under the multivariate normal, bottom row under $t_2$. Curves for the signed doubly ranked tests are in green, while curves for the functional sign tests are in blue (integral) and gold (sufficient statistic). Solid curves are for when $n = 60$, dashed curves for when $n = 30$, and dotted-dashed curves for when $n = 15$.\label{f:p80}}
\end{figure}

Tables~\ref{t:fdr} and~\ref{t:fdrdb} contain the estimated type I error from the general (under $\rho = 0.75$) and data-motivated simulation studies, respectively. Each value in these tables is calculated based on 10,000 simulated datasets using the nominal threshold of $\alpha = 0.05$. Bolded values indicate the test with the nearest-to-nominal type I error. As indicated by Table~\ref{t:fdr}, the signed doubly ranked test consistently has type I error that is the closest to nominal regardless of the sampling density ($S$), sample size ($n$), distribution of the process ($Z_{ijk}$), or the level of correlation ($\rho$). In only one instance does the integral-based functional sign test the have closest to nominal type I error: when $\rho = 0.75$, $S = 120$, $Z_{ijk} \sim t_2$, and $n = 30$. Outside of when $n = 30$, the functional sign tests have type I error well below nominal in the 0.02 to 0.03 range. The data-based simulation study produces similar results, see Table~\ref{t:fdrdb}. Thus when $\rho = 2/3$, $S = 80$, and we use FPCA SC instead of FACE for preprocessing, the signed doubly ranked test is still the closest to nominal regardless of the distribution of $Z_{ijk}$ and the sample size.

Figures~\ref{f:pone} and~\ref{f:ptwo} display the power curves from the general study under the two $\Delta$ types, respectively, and when $\rho = 0.75$. Values in these figures were generated based on 1,000 simulated datasets and the nominal threshold of $\alpha = 0.05$. The rows of each figure correspond to the distribution of the process while the columns vary by $S$. Across the three methods, power grows most noticeably as sample size increases. The sampling grid also plays a small role with denser grids producing slightly higher power. Power also tends to be slightly higher when the underlying process is normal, compared to when it comes from the $t_2$ distribution. Comparing between methods, the signed doubly ranked test consistently has the highest power; although when the distribution is $t_2$ in the densest setting, under the smallest sample size, and for $\Delta_1(s)$, the three methods before similarly. Comparing between the functional sign tests, the integral-based functional sign test tends to perform better than the sufficient statistic-based version. However, this advantage diminishes as $n$ increases. Similar graphics for when $\rho = 0.5$ are in Figures~\ref{f:pone50} and~\ref{f:ptwo50} where we observe similar trends as in the $\rho = 0.75$ case.

The power curves for the data-based setting are in Figure~\ref{f:p80}. For this figure, the columns correspond to type of $\Delta$ since the sampling grid is fixed at $S = 80$. In this study, we see similar results to the general study where the signed doubly ranked test outperforms the functional sign tests. Once again, the second best performing approach is the integral-based functional sign test. This data-based study uses a different FPCA approach when preprocessing: FPCA SC instead of FACE. These results suggest that the preprocessing step does not necessarily impact power.
	
\section{Reanalysis of the Flight Study}
\label{s:data}

\cite{Meyer2019} present a study of in-flight heart rate and hear rate variability in older and vulnerable patients. Groups of participant entered a hypobaric pressure chamber on two separate days; a flight day at altitude and a control day at sea level. On the flight day, the chamber was depressurized to 7,000 ft. altitude while on the control day, the chamber was turned on but no depressurized. Of the 42 participants recruited for the study, 34 had partially observed curves on both days. Missingness within curves was largely due to technical problems with the instrumentation during the experiment. The order in which the treatment was received was block randomized (by group) with 25 patients experiencing the flight condition on their first and 9 patients receiving the flight condition on their second day in the chamber. In the original analysis, \cite{Meyer2019} use a longitudinal difference-in-difference analysis to assess the changes in five outcomes pre-condition versus post---patients were in the chamber for 40 minutes prior to the chamber being turned on. The authors also adjust their analysis for the time of day the participants entered the chamber and a chamber familiarity effect. In this reanalysis, we focus only on 79 measurements taken every five minutes from the portion of the experiment when the chamber was on. Since we are using this data as an illustration, our analysis does not adjust for time of day, chamber familiarity, or the pre-condition period.

Table~\ref{t:results} contains the results of our analysis using the signed doubly ranked test and the both functional sign tests. We consider the same five outcomes that \cite{Meyer2019} examine: heart rate (HR), root mean square of successive RR interval differences (rMSSD), the standard deviation of normal-to-normal intervals (SDNN), high frequency power (HF), and low frequency power (LF). The latter four generally measure heart rate variability or HRV, although these metrics vary slightly in their interpretations; see \cite{Meyer2019} for additional details. Artificially induced changes to heart rate (increases) or HRV (decreases) can have potentially negative effects on heart health. Thus, we seek to see if there are difference over the duration of the simulated flight between conditions. The data we analyze from this study is presented in Figure~\ref{f:hr}.

\begin{table}
	\centering
	\caption{Reanalysis of flight data from \cite{Meyer2019} using functional sign and signed doubly ranked tests. Table values include the test statistic, either $W$ or $U$, and the corresponding $p$-value (parenthetically). SDRT denotes signed doubly ranked test, FST denotes functional sign test, Int. denotes integral, and Suff. denotes sufficient statistic.\label{t:results}}
	\begin{tabular}{lclclc}
		\hline
		\multirow{2}{*}{Outcome} & \multirow{2}{*}{SDRT} & & \multicolumn{3}{c}{FST}  \\
		\cline{4-6}
		 &  & & Int. & & Suff. \\
		\hline
		HR & 424 (0.0299) & & 23 (0.0576) & & 23 (0.0576)  \\
		rMSSD & 247 (0.3974) & & 14 (0.3915) & & 13 (0.2295)  \\
		SDNN &  220 (0.1906) & & 12 (0.1214) & & 13 (0.2295) \\
		HF & 240 (0.3342) & & 13 (0.2295) & & 12 (0.1214) \\
		LF & 264 (0.5772) & & 13 (0.2295) & & 13 (0.2295) \\
		\hline
	\end{tabular}
\end{table}

When examining the signed doubly ranked test results in Table~\ref{t:results}, we see that heart rate returns a test statistic of $W = 424$ with a corresponding $p$-value of $p = 0.0299$, giving evidence to suggest there is a difference in heart rate under the two conditions. The expected value of $W$'s distribution for $n = 34$ subjects is $\frac{n(n+1)}{4}$ or 297.5; this is the same expectation as in the univariate case \citep{Wilcoxon1945}. Since the observed $W$ is larger than this value, there are more positive signed ranks, on average. To obtain the difference, we take the flight curve minus the control curve. Thus, we have evidence to suggest that heart rate is higher under the flight condition than under the control condition. \cite{Meyer2019} come to the same conclusion in their adjusted analysis.

The analysis of the signed doubly ranked tests for the remaining HRV metrics does not provide any definite conclusions. The test statistics are all below 297.5 implying that, on average, there tend to be more negative signed ranks. However, the corresponding $p$-values range from 0.1906 to 0.5772. These results are also consistent with the findings of \cite{Meyer2019}. The test statistics for the functional sign tests follow the same directional pattern with $U$ being larger than expected for HR and lower than expected for the HRV metrics; the expected value under $H_0$ is 17 for the functional sign tests. In the case of heart rate, both functional sign tests return marginally significant---that is, $p$ between 0.05 and 0.10---results with the same test statistic: $U = 23$ (p = 0.0576). However, the corresponding $p$-values for HRV range from 0.1214 to 0.3915.

	\section{Discussion}
	\label{s:disc}
	
	In this manuscript, we propose two functional sign tests and a signed doubly ranked test for analyzing pairs of functional data. Prior work in this area was limited in scope and performed the basis transformation under the alternative, rather than under the null. As in the univariate case, we demonstrate that the signed rank test is more powerful than the functional sign tests. This can also be seen in our reanalysis of the heart rate data from the flight study. However, the comparison is not entirely apt as the relationships between the functional sign tests and the signed doubly ranked test are not as direct as in the univariate setting. The signed doubly ranked test does not take its signs from the functional sign tests but rather uses a summary of signed ranks at each $s$. This approach builds the test within the doubly ranked framework which allows it to maintain the null properties of the signed rank test at each step. Nevertheless, the signed doubly ranked test does exhibit good control of type I error while consistently maintaining the highest power.
	
	As is the case for all doubly ranked tests, the signed doubly ranked test is a global test that generally indicates a difference in the signed ranks which we can use to suggest a direction, on average, over the difference curve. The functional sign tests are also global test. These three tests cannot determine specific locations of differences when applied to the curves in their entirety. They can, however, indicate that the curves under one condition were typically higher (or lower) than the curves under the other condition. The signed doubly ranked test could be used as a primary analysis of paired functions or as part of an exploratory data analysis leading to more sophisticated modeling of the curves using other functional data techniques. While we employ both FACE and FPCA SC in the preprocessing step, one can use any basis representation of the functional curves or use the raw data themselves as our theoretical results do not rely on the preprocessing step. They simply assume that $D_i(s)$ comes from a difference curve, however the user defines it.

\bibliographystyle{abbrvnat} % Style BST file (imsart-number.bst or imsart-nameyear.bst)
\bibliography{fullbib.bib}

\end{document}